\documentclass[a4paper,reqno]{amsart}

    \usepackage[english]{babel} 

    \usepackage{graphicx}
    \usepackage[]{color}  
    \usepackage[dvipsnames]{xcolor} 
    \usepackage{csquotes} 
    \usepackage{appendix} 
    \usepackage[a4paper]{geometry} 
    \usepackage{cancel} 

    \usepackage{amsmath,amssymb,amsfonts,amsthm} 
    \usepackage{mathtools,todonotes} 
    \usepackage{tikz-cd}  
    \usepackage[all,cmtip]{xy} 
    \usepackage[bbgreekl]{mathbbol} 
    \usepackage{array} 

    \usepackage{hyperref} 

\theoremstyle{definition} 
    \newtheorem{definition}{Definition}

\theoremstyle{plain} 
    \newtheorem{theorem}[definition]{Theorem}
    \newtheorem{proposition}[definition]{Proposition}
    \newtheorem{lemma}[definition]{Lemma}

\theoremstyle{remark} 
    \newtheorem{remark}[definition]{Remark}

\usepackage[bibstyle=alphabetic,citestyle=alphabetic,useprefix,giveninits=true, sorting=ynt, sortcites, minbibnames=99,maxbibnames=99,backend=biber]{biblatex}  
\renewbibmacro{in:}{} 
\bibliography{bibliography_corner}  
\DeclareSourcemap{
  \maps[datatype=bibtex]{
    \map[overwrite]{ 
      \step[fieldsource=doi, final]
      \step[fieldset=url, null]
      \step[fieldset=eprint, null]
      }
     \map[overwrite]{ 
      \step[fieldsource=eprint, final]
      \step[fieldset=pages, null]
      \step[fieldset=eid, null]
      \step[fieldset=journal, null]
    }  
  }
}

    \renewcommand{\Tilde}{\widetilde}   
    \newcommand{\tc}{\widetilde{c}}
    \newcommand{\tom}{\widetilde{\omega}}
    
    \newcommand{\te}{\widetilde{e}}
    \newcommand{\txi}[1]{\widetilde{\xi}^{#1}}

    \DeclareMathOperator{\Ima}{Im}

    \newcommand{\Ker}[1]{\mathrm{Ker}{(#1)}}

    \newcommand{\qsp}[2]{\,\ensuremath{\raise.5ex\hbox{$#1$}\big\slash\raise-.5ex\hbox{$#2$}}}

    \newcommand{\dd}{\mathrm{d}}
    \newcommand{\ndash}{\nobreakdash-\hspace{0pt}}

    \newcommand{\filt}[1]{(#1)}

\newcommand{\bbR}{\mathbb{R}}
\newcommand{\bbZ}{\mathbb{Z}}
\newcommand{\frg}{\mathfrak{g}}
\newcommand{\frh}{\mathfrak{h}}
\newcommand{\frp}{\mathfrak{p}}
\newcommand{\frso}{\mathfrak{so}}
\newcommand{\de}{\partial}
\newcommand{\calA}{\mathcal{A}}
\newcommand{\calB}{\mathcal{B}}
\newcommand{\calF}{\mathcal{F}}
\newcommand{\calG}{\mathcal{G}}
\newcommand{\calN}{\mathcal{N}}

\newcommand{\bA}{\mathbf{A}}
\newcommand{\bE}{\mathbf{E}}

\begin{document}
\title[Corner Structure of GR]{Corner Structure of Four-Dimensional General Relativity in the Coframe Formalism}
\author[G. Canepa]{Giovanni Canepa}
\address{Fakult\"at f\"ur Physik, Universit\"at Wien, Boltzmanngasse 5, 1090 Wien, \"Osterreich}
\email{giovanni.canepa.math@gmail.com}

\author[A. S. Cattaneo]{Alberto S. Cattaneo}
\address{Institut f\"ur Mathematik, Universit\"at Z\"urich, Winterthurerstrasse 190, 8057 Z\"urich, Switzerland}
\email{cattaneo@math.uzh.ch}

\thanks{A. S. C. acknowledges partial support of SNF Grant No.\ 200020 192080 and of the Simons Collaboration on Global Categorical Symmetries. G. C. acknowledges partial support of SNF Grant No P500PT\textunderscore203085. This research was (partly) supported by the NCCR SwissMAP, funded by the Swiss National Science Foundation.}

\maketitle
\begin{abstract}
    This note describes a local Poisson structure (up to homotopy) associated with corners in four-dimensional gravity in the coframe (Palatini--Cartan) formalism. This is achieved through the use of the BFV formalism.  The corner structure contains in particular an Atiyah algebroid that couples the internal symmetries to diffeomorphisms. The relation with $BF$ theory is also described.
\end{abstract}

\setcounter{tocdepth}{3}
\tableofcontents

\section{Introduction}
The goal of this paper is to describe the Poisson structures (up to homotopy) that arise on two-dimensional corners of four-dimensional gravity in the coframe (Palatini--Cartan) formalism. 

{}From a more general perspective, one expects quantum field theory on a cylinder to describe the quantum evolution of a system described by a Hilbert space attached to a boundary component. If the boundary has itself a boundary---a corner for space--time---, the Hilbert space is expected to be a representation of some algebra associated with the corner. 
A standard example where this picture is considered is that of the vertex operator algebra arising from a punctured two-dimensional boundary.

At the classical level, one then expects a symplectic manifold to be associated with a boundary and a Poisson manifold to be associated with a corner. This picture is, however, problematic, since the constructions typically involve singular quotients.

A more suitable picture, which we use in this paper, is that of the Batalin--Fradkin--Vilkovisky (BFV) formalism \cite{BV1,BV2,BV3}, which replaces a (possibly singular) symplectic quotient by a cohomological resolution: namely, one extends the space of boundary fields to a superspace with additional structure (a symplectic structure---the BFV form---together with a Hamiltonian vector field that squares to zero---the BRST operator). 

An added bonus of this formalism is that it naturally produces a structure on the corners \cite{CMR2012,CMR2012b} which, upon choosing a ``polarization'' (i.e., a choice of a foliation by Lagrangian submanifolds) is associated with a Poisson structure (up to homotopy). 

We recall this construction, together with background material, in the first part of Section~\ref{s:detodede}, whereas in its second part we apply it to some 
instructive examples (Yang--Mills, Chern--Simons, and, notably, 4D $BF$ theory).

In Section~\ref{s:BFVPC} we recall the BFV formulation of 4D Palatini--Cartan theory \cite{CCS2020}, and
in Section~\ref{s:cornerind} we apply the aforementioned construction for corners and observe that it is singular. Nonetheless, it is possible to study and describe, in Section~\ref{s:Pstru}, a naturally associated local Poisson algebra up to homotopy. {This algebra is actually generated through Poisson brackets and a differential by the observables}
\[
J_\phi=\frac12\int_\Gamma \phi ee,
\]
where $\Gamma$ is the two-dimensional corner, $e$ is the coframe (tetrad) field (restricted to the corner), and $\phi$ is an $\mathfrak{so}(3,1)$-valued test function (Lie algebra pairing is tacitly understood in the notation).

These particular observables are reminiscent of the area observables considered in loop quantum gravity (see, e.g., \cite{Rovelli} and references therein), where, however, $\Gamma$ is a closed surface inside the boundary instead of being a corner (and Ashtekar $\mathfrak{su}(2)$ variables are used instead). In particular, our observable has a similar form of the variable conjugate to the holonomy of the connection in loop quantum gravity. Namely, $J_\phi$ is the same as the variable (7.7) introduced in \cite{Rovelli2}.
The corner structure leads to the Poisson bracket $\{J_{\phi_1},J_{\phi_2}\}_\text{corner}=J_{[\phi_1,\phi_2]}$, which is in line with the Poisson bracket of area observables, although we use here the Poisson bracket associated with the corner instead of that associated with the boundary\footnote{More precisely, the corner and the boundary observables live on different spaces. The restriction map to the corner yields, however, a map from the boundary fields to the corner fields. The fact that the Poisson brackets among the $J_\phi$s agree when calculated with respect to the boundary or the corner structure simply means that the restriction map is, at least as far as these observables are considered, a Poisson map.} and, unlike in \cite{CP2017}, no regularization is required in this context.

The above observables retain information of the internal $\mathfrak{so}(3,1)$ symmetry of Palatini--Cartan gravity. The other observables they generate, through the differential in the homotopy Poisson algebra, contain information about tangential and transversal vector fields encoding the diffeomorphism symmetry as well.

An interesting fact, which deserves further investigation, is that this corner structure actually turns out to be the corner structure for four-dimensional $BF$ theory restricted to a submanifold of fields. 

In order to understand better the algebra found, it is useful to consider some particular cases. In Section~\ref{s:simplified} we describe two possible restrictions of the general theory, called \textit{constrained} and \textit{tangent} theory, and produce a better description of a restricted version of the aforementioned local Poisson algebra up to homotopy. In the first (Section~\ref{s:constrained_Poisson_algebra}) we impose some ad hoc constraints that do not modify the classical structure of the theory, while in the second (Section~\ref{s:Poissontangent}))
we essentially freeze the generators of transversal diffeomorphisms. In the tangent theory, the associated Poisson manifold turns out to be a Poisson submanifold of the dual space of sections of an Atiyah algebroid associated with the corner (Section~\ref{s:atiyah_tangent}). We briefly discuss the quantization when the corner is a sphere and the fields are assumed to be constant---a situation that is relevant in the case of a punctured boundary (Section~\ref{s:quantization_tangent}).

These results are of course expected to be related to the BMS group \cite{BBM62,Sachs1962,P1963,Strominger2014, HLMS2015} at infinity, which has been extensively studied (see, e.g., \cite{BT2011,FOPS21} and references therein). Furthermore, we expect some connections between the results of this paper and those derived recently in the context of the covariant phase space method (see, e.g., \cite{MS2021, OS2020, CL2021, FGP2020} and references therein). This will be object of future work. 
The difference with our approach is that we assume the boundary of space--time to be a compact manifold with boundary. For a noncompact manifold, one should instead choose an appropriate compactification, related to the chosen  asymptotic conditions for the fields. 
 Given the differences in the formalism, the connection between these structures is not fully understood; hence, we plan to explore these topics in a forthcoming work. A similar Poisson structure on corners, in the case when a momentum map is available, can be found in \cite{RS2023}.

Some of the results in this paper (in particular Sections \ref{s:BFFV_Corner}, \ref{sec:constrainedtheory} except Section \ref{s:constrained_Poisson_algebra},  and the first part of \ref{sec:tangenttheory} until Section \ref{s:Poissontangent}) first appeared in the PhD thesis of the first author\cite{C2021} available online but not published in a peer-reviewed journal.

\subsection*{Acknowledgments} We thank M. Schiavina and S. Speziale for the fruitful discussions that we had during the preparation of this article. A.C. also thanks P. Xu and T. Voronov for their illuminating suggestions.

\section{Preliminaries and relevant constructions}\label{s:detodede}
In this section, we review how the BFV formalism is used to describe coisotropic reduction, which is relevant for the boundary structure of a field theory, how the BF$^2$V formalism is used to describe Poisson structures (possibly up to homotopy),
which is relevant for the corner structure of a field theory, and how the two may be related.

\begin{remark}
We group here some references for this section, not to interrupt the flow of the following. For Poisson and symplectic structures, see, e.g.,  \cite{BW1997}. The notion of coisotrope was introduced in \cite{W1988}. The notion of derived bracket was introduced in \cite{K1996} and generalized in \cite{V2005,V2004}.
The notion of BF$^m$V structures and their mutual relations, in particular arising from relaxed structures, was introduced in \cite{CMR2012,CMR2012b}, although not with this name; note that there is a parallel story developed in derived symplectic geometry, see \cite{Cal15, CPTVV,Saf20} and references therein. 
The existence of BFV structures associated with coisotropic submanifolds is discussed in \cite{Stasheff1997,Schaetz:2008,SchaetzTH,FK2012}.
\end{remark}

\subsection{Background notions}
We start recalling some important preliminaries.
\subsubsection{Poisson and symplectic structures}
\begin{definition}
    A Poisson algebra is a pair $(A, \{\ ,\ \})$ where $A$ is a commutative algebra (for our applications always over $\bbR$) and $\{\ ,\ \}$ is a bilinear, skew-symmetric operation on $A$ which is a derivation w.r.t.\ each argument (Leibniz rule)---i.e., a biderivation---and satisfies the Jacobi identity. The operation $\{\ ,\ \}$ is called a Poisson bracket.
\end{definition}
 The simplest example of a Poisson algebra is any algebra with the zero Poisson bracket. Another interesting example is the symmetric algebra $S(\frg)$ of a Lie algebra $\frg$, where the Lie bracket is extended by the Leibniz rule. Symplectic manifolds also produce Poisson algebras, as we recall below.

\begin{definition}
    A Poisson manifold is a pair $(M,\{\ ,\ \})$ where $M$ is a smooth manifold and $\{\ ,\ \}$ is a Poisson bracket on $C^\infty(M)$. 
\end{definition}
Again we have the simplest example of the zero Poisson bracket. The dual $\frg^*$ of a finite-dimensional Lie algebra $\frg$ is also an example, where the Poisson bracket on $S(\frg)$, now viewed as the algebra of polynomial functions on $\frg^*$, is extended to the whole $C^\infty(\frg^*)$. 

A biderivation $\{\ ,\ \}$ on a smooth manifold $M$ is always determined by a bivector field $\pi$ via
$\{f,g\}=-\pi(\dd f,\dd g)$. If we denote by $[\ ,\ ]$ the Schouten bracket of multivector fields, the Jacobi identity for the bracket is equivalent to the Maurer--Cartan equation $[\pi,\pi]=0$. In this case, $\pi$ is called a Poisson bivector field.
Moreover, we can also write $\{f,g\}=[[\pi,f],g]$, which is an example of derived bracket, on which we will elaborate below. In the trivial case, $\pi$ is the zero bivector field. In the case of the dual of a Lie algebra $\frg$, we have $\pi^{ij}=-f^{ij}_k x^k$, where the $f^{ij}_k$s are the structure constant of $\frg$ in some basis and the $x^k$s are the coordinate on $\frg^*$ w.r.t.\ the same basis.
\begin{definition}
    A symplectic manifold is a pair $(M,\varpi)$ where $M$ is a smooth manifold and $\varpi$ is a closed nondegenerate two-form on $M$. If $M$ is infinite dimensional, we require only weak nondegeneracy, namely, that at every point $x$
\[
\varpi_x(v,w) = 0\ \forall v\in T_xM \implies w=0.
\]
\end{definition}

This condition implies that a function $f$ has at most one Hamiltonian vector field $X_f$: $\iota_{X_f}\varpi = \dd f$.
We say that a function is Hamiltonian if it has a Hamiltonian vector field and denote the space of such functions $C^\infty(M)_\text{Hamiltonian}$. The Poisson bracket of two Hamiltonian functions $f$ and $g$, with Hamiltonian vector fields denoted $X_f$ and $X_g$, respectively, is defined as
\[
\{f,g\} := X_f(g) = \iota_{X_f}\iota_{X_g}\varpi.
\]
It is a Poisson bracket on $C^\infty(M)_\text{Hamiltonian}$. If $M$ is finite dimensional, then $(M,\{\ ,\ \})$ is a Poisson manifold; the corresponding Poisson bivector field is the inverse of the symplectic structure.

\begin{remark}\label{r:Hamiltoniandeg}
The above can be generalized to the case when we drop the nondegeneracy condition. In this case, we say that a vector field $X$ is in the kernel of $\varpi$ if $\iota_X\varpi=0$. We call a function $f$ invariant if $X(f)=0$ for every $X$ in the kernel of $\varpi$. We call, as before, $f$ Hamiltonian if it possesses a Hamiltonian vector field $X_f$: $\iota_{X_f}\varpi = \dd f$. Note that in general the Hamiltonian vector field is no longer unique. A Hamiltonian function is automatically invariant. The action of a Hamiltonian function $f$ on an invariant function $g$ is defined as $\{f,g\} := X_f(g)$, where it does not matter which Hamiltonian vector field we take, and produces an invariant function. If also $g$ is Hamiltonian, then the result is Hamiltonian as well, and $\{\ ,\ \}$  is a Poisson bracket on $C^\infty(M)_\text{Hamiltonian}$.
\end{remark}

\subsubsection{Coisotropic submanifolds and reduction}
\begin{definition}
    A coisotrope in a Poisson algebra $(A, \{\ ,\ \})$ is an ideal $I$ in the commutative algebra $A$ which satisfies $\{I,I\}\subseteq I$: i.e., $I$ is a Lie subalgebra of $(A, \{\ ,\ \})$.
\end{definition}
 Note that $I$ naturally acts on the commutative algebra $A/I$ via the bracket. We also have $(A/I)^I= N(I)/I$, where $N(I):=\{a\in A\ |\ \{a,I\}\subseteq I\}$ is the Lie normalizer of $I$ in $A$. The latter description shows that  $\underline A_I:=(A/I)^I=N(I)/I$ is a Poisson algebra, called the reduction of $A$ w.r.t.\ to $I$.

\begin{definition}
    A coisotropic submanifold of a Poisson manifold $(M,\{\ ,\ \})$ is a submanifold\footnote{We only consider closed submanifolds.} $C$ of $M$ such that its vanishing ideal $I$ is a coisotrope in $(C^\infty(M),\{\ ,\ \})$.
\end{definition}

\begin{remark}
    If $C$ is the zero locus of constraints $\phi_i$, the latter condition is equivalent to having $\{\phi_i,\phi_j\}=f_{ij}^k \phi_k$, where summation over repeated indices is understood and the $f_{ij}^k$s are functions, called the structure functions. Constraints satisfying this condition are called first class in Dirac's terminology.
\end{remark}

If $M$ is a finite-dimensional symplectic manifold, then this definition of coisotropic submanifold is equivalent to the geometric one that,  for every $x\in C$, the subspace $T_xC$ be coisotropic, i.e., $(T_xC)^\perp\subseteq T_xC$,\footnote{The orthogonal space is taken w.r.t. the symplectic form, i.e.,
\[
(T_xC)^\perp=\{ v\in T_xM \ |\ \varpi_x(v,w)=0\ \forall w\in T_xC\}.
\]} for every $x\in C$.
The Hamiltonian vector fields of elements of the vanishing ideal span the involutive distribution $(TC)^\perp$. 
\begin{proposition}
    If the quotient space $\underline C$ has a smooth manifold structure for which the projection $\pi\colon C\to\underline C$ is a smooth submersion, then $\underline C$ is endowed with a unique symplectic structure $\underline\varpi$ such that
$\pi^*\underline\varpi=\iota^*\varpi$, where $\iota\colon C\to M$ is the inclusion map. The pair $(\underline C,\underline\varpi)$ is called the symplectic reduction of $C$.
In this case, the resulting Poisson algebra $C^\infty(\underline C)$ is the reduction $\underline A_I$ described above.
\end{proposition}

If $M$ is an infinite-dimensional symplectic manifold, there are inequivalent ways of defining a coisotropic submanifold. In this paper, we will stick to the algebraic definition. More precisely, we assume that the vanishing ideal $I$ is generated by
its Hamiltonian part $I_\text{Hamiltonian}:=I\cap C^\infty(M)_\text{Hamiltonian}$ and that $I_\text{Hamiltonian}$ is a coisotrope in $C^\infty(M)_\text{Hamiltonian}$.

\begin{remark}
The importance of coisotropic submanifolds in field theory is related to the problem of finding the correct space of initial conditions for the Cauchy problem. Indeed, the coisotropic submanifold $C$ arises as a submanifold of the space of boundary fields with the constraints determined by the Euler--Lagrange equations that do not involve transversal derivatives.
In case this construction arises from the Hamiltonian description associated with a Cauchy surface, the reduced phase space, i.e., the reduction $\underline C$ of 
$C$,
is the correct space of initial conditions for the Cauchy problem. 
\end{remark}

\subsubsection{The graded case: \texorpdfstring{BF$^m$V}{BFmV} structures}
All the above can be extended to the world of graded algebras and graded manifolds (supermanifolds with an additional $\bbZ$\ndash grading on the local coordinates). Note that we assume both a grading and a parity, the latter being responsible for the sign rules. In all the examples in this paper they are related, with the parity being the grading modulo two.

\begin{definition}
    A graded Poisson algebra is a pair $(A, \{\ ,\ \})$ where $A$ is a graded commutative algebra  
and $\{\ ,\ \}$ is a bilinear, graded skew-symmetric operation on $A$ which is a graded derivation w.r.t.\ each argument (graded Leibniz rule) and satisfies the graded Jacobi identity.
\end{definition}
 It is important to notice that the grading of the bracket may be a shifted grading w.r.t.\ the original one. 

An even bracket of degree $0$---the straightforward generalization from the ungraded case---is also known as a BFV bracket. An odd bracket of degree $+1$ is also known as a BV bracket. We will call an odd bracket of degree $-1$ a BF$^2$V bracket.

\begin{definition}
    An $n$-graded symplectic manifold is a pair $(M,\varpi)$ where $M$ is a graded manifold and $\varpi$ is a closed nondegenerate two-form on $M$ of homogenous degree $n$ and parity $n\!\!\mod 2$. It defines a graded Poisson algebra structure on $C^\infty(M)_\text{Hamiltonian}$ with bracket of degree $-n$.
\end{definition}

An additional structure, important for the following, is that of cohomological vector field on a graded manifold $M$.
This is an odd vector field $Q$ of degree $+1$ satisfying $[Q,Q]=0$. Note that $Q$ defines a differential on $C^\infty(M)$. For this reason, the pair $(M,Q)$ is called a differential graded manifold (shortly, a dg manifold).

\begin{definition}
    A dg manifold with a compatible symplectic structure, i.e., with $L_Q\varpi=0$, is called a differential graded symplectic manifold (shortly, a dg symplectic manifold).
\end{definition}
 We will always assume that $Q$ is Hamiltonian, namely, that there is an
$S\in C^\infty(M)_\text{Hamiltonian}$ such that $\iota_Q\varpi = \dd S$ and $\{S,S\}=0$ (the master equation).\footnote{For most choices of $n$, the existence of $S$ is guaranteed and the condition $\{S,S\}=0$ is equivalent to $[Q,Q]=0$.}
If $\varpi$ has degree $n$, then $S$ has degree $m=n+1$. In this case, we call the triple $(M,\varpi,S)$ a BF$^m$V manifold.

\begin{remark}
BV manifolds arise in field theories as a generalization of the BRST formalism to discuss independence of gauge-fixing in the perturbative functional-integral quantization; we will not address this issue in this paper. BFV manifolds are used to give a cohomological description of reduced phase spaces. BF$^2$V manifolds describe Poisson structures (up to homotopy). We will recall these two constructions  in Sections \ref{s:BFV_formalism} and \ref{sec:BF2Vstructure}, respectively.
\end{remark}

\subsubsection{Relaxed and induced structures}\label{s:relind}
The above may be generalized by dropping the master equation, the condition that $\varpi$ is nondegenerate, and the strict relation among $(Q,\varpi,S)$. Namely, we only assume that $\varpi$ is a closed two-form on $M$ of homogenous degree $(m-1)$ and parity $(m-1)\!\!\mod 2$ and that $Q$ is a cohomological vector field: we call this a relaxed BF$^m$V structure.
We define $\Tilde\alpha:=\iota_Q\varpi-\dd S$ and $\Tilde\varpi=\dd\Tilde\alpha$. It turns out that $Q$ and $\Tilde\varpi$ are compatible, i.e., $L_Q\Tilde\varpi=0$. We actually assume the slightly stronger condition
$\iota_Q\Tilde\varpi=\dd\Tilde S$ for some function $\Tilde S$. 
One can also show the useful identity $\frac12\iota_Q\iota_Q\varpi=\Tilde S$, called the modified master equation.
We call the triple $(M,\Tilde\varpi,\Tilde S)$, or any of its partial reductions by an integrable subdistribution of the kernel of $\Tilde\varpi$, a pre--BF$^{m+1}$V manifold. If the whole reduction by the kernel is smooth, it is then a BF$^{m+1}$V manifold as defined above. In this case, we say that the relaxed BF$^m$V structure is $1$-extendable.

\begin{remark}
In the case of field theory, we always assume locality. Namely, $M$ is a space locally modeled on sections, the fields, of a vector bundle over some closed manifold $\Sigma$, and the structures $(Q,\varpi,S)$ are integrals over $\Sigma$ of densities defined, at each point, in terms of jets of the fields. The relaxed structure typically arises when one extends the strict structure to a manifold with boundary,\footnote{\label{f:compact}Typically, we assume compactness. Otherwise, one has to specify appropriate vanishing conditions on the fields.} by taking the same triple $(Q,\varpi,S)$. In this case, the ``error term'' $\Tilde\alpha$ arises by integration by parts and is concentrated on $\partial\Sigma$. Modding out by (part of) the kernel of $\Tilde\varpi$ then yields a (pre--)BF$^{m+1}$V structure depending on jets of the fields restricted to $\partial\Sigma$.
\end{remark}

\subsubsection{The BFV formalism}\label{s:BFV_formalism}
If $(M,\varpi,S)$ is a BFV manifold, then the zeroth cohomology group $H^0_Q(C^\infty(M)_\text{Hamiltonian})$ is a Poisson algebra.\footnote{Recall that $Q$, the Hamiltonian vector field of $S$, is a differential on the algebra of Hamiltonian functions.}
Namely, if $[f]$ and $[g]$ are cohomology classes, we define $\{[f],[g]\}:=[\{f,g\}]$. This Poisson algebra is understood as the algebra of function of a would-be symplectic reduction.

This is justified by the BFV construction. Namely, one starts with a symplectic manifold $(M_0,\varpi_0)$ and a coisotropic submanifold $C$ of $M_0$. One can then associate with it a BFV manifold $(M,\varpi,S)$ that contains $(M_0,\varpi_0)$ as its degree zero part and such that $C$ is recovered as the intersection of $M_0$ with the critical locus of $S$. (This construction works in general if $M_0$ is finite dimensional; in the infinite-dimensional case, it works at least when $C$ is given by global constraints.) For example, if $M$ is finite dimensional and $C$ is locally defined by constraints $\phi_i$, then in local coordinates we have $S=c^i\phi_i + \cdots$, where the $c^i$s are the coordinates of degree $+1$ and the dots are in the ideal generated by the coordinates of degree $-1$.
The dots here have to be added to ensure that the master equation is satisfied.

If $C$ has a smooth reduction $\underline C$, then $H^0_Q(C^\infty(M)_\text{Hamiltonian})$ is isomorphic, as a Poisson algebra, to $C^\infty(\underline C)$. In general, one views $(M,\varpi,S)$ as a good replacement (a cohomological resolution) for the reduction of $C$.

\subsection{\texorpdfstring{$P_\infty$}{P infinity} structures from the \texorpdfstring{BF$^2$V}{BF2V}  formalism}\label{sec:BF2Vstructure}
In this case, $\varpi$ is an odd symplectic form of degree $+1$. 
We start with the finite-dimensional case. One then has that
$(M,\varpi)$ is always symplectomorphic to a shifted cotangent bundle $T^*[1]N$, with canonical symplectic structure, for some graded manifold $N$ (with this notation we mean that the fiber coordinates of $T^*N$ are assigned opposite parity and degree shifted by one w.r.t.\ the natural ones). We call this choice of $N$ a polarization. Note that the Poisson algebra of functions on $T^*[1]N$ can be canonically identified with the algebra of multivector fields on $N$ with the Schouten bracket. The function $S$, of degree $+2$, then corresponds to a linear combination $\pi=\pi_0+\pi_1+\pi_2+\cdots$, where $\pi_i$ is an $i$-vector field of degree $2-i$ on $N$. The master equation $\{S,S\}=0$ corresponds to the equations
\begin{align*}
[\pi_0,\pi_1]&=0,\\
[\pi_0,\pi_2]+\frac12[\pi_1,\pi_1]&=0,\\
[\pi_0,\pi_3]+[\pi_1,\pi_2]&=0,\\
[\pi_0,\pi_4]+[\pi_1,\pi_3]+\frac12[\pi_2,\pi_2]&=0,\\
\dots
\end{align*}

We start from the simpler case when $N$ has only coordinates in degree zero (this is possible only if $M$ has only coordinates in degree zero and one). In this case, $\pi=\pi_2$ and $[\pi_2,\pi_2]=0$, so $\pi$ is a Poisson structure on $N$. Algebraically, we can get the corresponding Poisson algebra as the algebra $C^\infty_0(T^*[1]N)$ of functions on $T^*[1]N$ of degree zero with Poisson bracket $\{f,g\}_2=[[\pi,f],g]$.

In the general case, $\pi$ is called a $P_\infty$ structure on $N$ (this stands for Poisson structure up to coherent homotopies). This structure is called curved if $\pi_0\not=0$.
The $\pi_i$s, applied to the differentials of $i$ functions on $N$, define multibrackets $\{\ \}_i$ on $C^\infty(N)$ 
which in turn define a (curved) L$_\infty$-algebra. Moreover, they are graded derivations w.r.t.\ each argument. The multibrackets may also be defined as derived brackets
\[
\{f_1,\dots,f_i\}_i = [[[[\cdots[\pi_i,f_1],f_2],\dots],f_i] = P[[[[[\cdots[\pi,f_1],f_2],\dots],f_i],
\]
where $P$ is the projection from multivector fields to functions. In particular, we have
\begin{align*}
\{\}_0 &= \pi_0,\\
\{f\}_1 &= \pi_1(f),\\
\{f,g\}_2 &=[[\pi_2,f],g].
\end{align*}

We will call these brackets respectively the nullary, unary and binary operations or, equivalently, the 0-bracket, 1-bracket and 2-bracket.

\subsubsection{Generalizations}\label{s:generalization_Poisson}
The above structure may be generalized as follows. Suppose we have a splitting $\mathfrak{a}=\frp\oplus\frh$ of an odd Poisson algebra $\mathfrak{a}$ (e.g., $C^\infty(M)$) into Poisson subalgebras with $\frh$ abelian (i.e., $\frp\cdot\frp\subseteq\frp$, $\frh\cdot\frh\subseteq\frh$, $\{\frp,\frp\}\subseteq\frp$,
$\{\frh,\frh\}=0$). Let $P$ be the projection $\mathfrak{a}\to\frh$. If $S\in\mathfrak{a}$ satisfies the master equation $\{S,S\}=0$, then the multibrackets
\[
\{f_1,\dots,f_i\}_i := P\{\cdots\{S,f_1\},f_2\},\dots\},f_i\}
\]
make $\frh$ into a $P_\infty$ algebra. The previous case consisted in considering $\mathfrak{a}=C^\infty(T^*[1]N)$ and taking $\frp$ as the multivector fields on $N$ of multivector degree larger than zero and $\frh$ as the functions on $N$; note that, in this case, $\frh$ is maximal as an abelian subalgebra. We call the more general choice of $(\frp,\frh)$ a weak polarization.

\begin{remark}\label{r:Poissondegenerateform}
The algebraic construction makes sense also if $\varpi$ is degenerate. In this case, we consider a splitting, with the above properties, of the $-1$-Poisson algebra of Hamiltonian functions:
$C^\infty_\text{Hamiltonian}(M)=\frp\oplus\frh$.
\end{remark}

\begin{remark}
An important case is when $\varpi$ is degenerate but its kernel has constant rank (i.e., the dimension of the kernel of $\varpi_x$ is the same for all $x\in M$). In this case, one calls it a presymplectic form. Note that the kernel is also involutive. If the quotient space of $M$ by the kernel has a smooth structure, it is then symplectic, so it can be identified with some $T^*[1]N$. We can then take $\frh=p^* C^\infty(N)$, where $p$ denotes the projection $M\to T^*[1]N$.
\end{remark}

\begin{remark}
More generally, we can take the quotient of $M$ by an involutive subdistribution of constant rank of the kernel of $\varpi$. If the quotient $\underline M$ has a smooth structure and $p$ denotes the projection from $M$ to $\underline M$, then we can take $\frh=p^* \frh'$, where $C^\infty_\text{Hamiltonian}(\underline M)=\frp'\oplus\frh'$ is a splitting as above.
\end{remark}

Let us now turn to the infinite-dimensional case. The first remark is that in this case, $M$ is symplectomorphic to a symplectic subbundle of $T^*[1]N$, for some infinite-dimensional graded manifold $N$. The only difference with the finite-dimensional case is that now not every function is Hamiltonian. We can anyway define the derived brackets, as above, on
$C^\infty_\text{Hamiltonian}(N):=C^\infty(N)\cap C^\infty_\text{Hamiltonian}(M)$. The algebraic version for weak polarizations and its extension to the degenerate case works verbatim as above.

\section{Corner structures of field theories}

In this section we consider some illustrating examples of BFV and BF$^2$V structures in field theory (electromagnetism, Yang--Mills theory, Chern--Simons theory, $BF$ theory). In particular, the example of $BF$ theory is preliminary to our discussion of these structures in gravity.

\begin{remark}
From here on we denote the differential on a space of fields by $\delta$, reserving the notation $\dd$ to
the de~Rham differentials on the underlying manifolds.
Furthermore, we will denote with an apex $\partial$ all the quantities with fields defined on $\Sigma$ and with an apex $\partial\partial$ all the quantities with fields defined on $\partial\Sigma$. This notation is chosen in order to make contact with the one used in many previous articles. This is due to the fact that often the BFV theory can be induced from a BV theory when $\Sigma$ is considered as a boundary of a manifold $M$.
\end{remark}

\subsection{Electromagnetism}
To warm up, we start with the simple example of electromagnetism in $d+1$ dimensions. In the Hamiltonian formalism, we then consider a $d$-dimensional Riemannian closed\footnote{Later, we will allow $\Sigma$ to be with boundary, but for simplicity we keep assuming compactness; see also footnote~\ref{f:compact} on page~\pageref{f:compact}.} manifold $(\Sigma,g)$, which for simplicity we assume to be oriented. The fields are the vector potential $\bA$ and the electric field $\bE$ with symplectic structure $\varpi^{\partial}_0=\int_\Sigma \delta \bA\cdot\delta \bE\,\sqrt{\det g}$, where $\cdot$ denotes the inner product defined by the Riemannian metric $g$ and $\sqrt{\det g}$ is the corresponding canonical density.

The constraints are given by the Gauss law $\operatorname{div} \bE = 0$. To implement the BFV formalism, we then have to introduce a ghost $c\in C^\infty(\Sigma)[1]$ and its conjugate momentum $b\in\Omega^d(\Sigma)[-1]$. 
We then have the BFV symplectic form

\[
\varpi^{\partial}=\int_\Sigma (\delta \bA\cdot\delta \bE\,\sqrt{\det g}+\delta b\,\delta c)
\]
and the BFV action
\[
S^{\partial} = \int_\Sigma c\,\operatorname{div} \bE\,\sqrt{\det g}.
\]
The variation of $S^{\partial}$ is
\[
\delta S^{\partial} = \int_\Sigma (\delta c\,\operatorname{div} \bE - c\,\operatorname{div} \delta \bE)\,\sqrt{\det g} =
\int_\Sigma (\delta c\,\operatorname{div} \bE + \operatorname{grad}c\cdot \delta\bE) \,\sqrt{\det g},
\]
which shows that $S^{\partial}$ is Hamiltonian, $\iota_{Q^{\partial}}\varpi^{\partial}=\delta S^{\partial}$, with $Q^{\partial}$ given by
\[
Q^{\partial}\bA = \operatorname{grad}c,\quad Q^{\partial}\bE \sqrt{\det g}=0,\quad Q^{\partial} b= \operatorname{div} \bE,\quad Q^{\partial} c=0.
\]
One can then see that the cohomology in degree zero consists of functionals of $\bA$ and $\bE$, modulo the ideal generated by $\operatorname{div} \bE$, that are gauge invariant. This is correctly the algebra of functions of the reduction of $C=\{(\bA,\bE)\ |\ \operatorname{div} \bE = 0\}$.

If $\Sigma$ has a boundary, we instead get
\[
\delta S^{\partial}
=
\int_\Sigma (\delta c\,\operatorname{div} \bE + \operatorname{grad}c\cdot \bE) \,\sqrt{\det g} + \int_{\de\Sigma} c\,\delta E_n \,\sqrt{\det g_{|_{\de\Sigma}}},
\]
where $E_n$ is the transversal component of $\bE$. This fits with the BFV-BF$^2$V prescription $\iota_{Q^{\partial}}\varpi^{\partial}=\delta S^{\partial}+\Tilde\alpha^{\partial}$ with
$\Tilde\alpha^{\partial}=\int_{\de\Sigma} c\,\delta E_n \,\sqrt{\det g_{|_{\de\Sigma}}}$. As $\Tilde\varpi^{\partial}=\delta\Tilde\alpha^{\partial}$ only depends on $c$ and on $E_n$ on $\de\Sigma$, we get the reduced space of fields 
$\calF_{\de\Sigma}=\{(c,E_n)\in C^\infty(\de\Sigma)[1]\oplus C^\infty(\de\Sigma)\}$ with BF$^2$V symplectic structure
\[
\varpi^{\partial\partial} = \int_{\de\Sigma} \delta c\,\delta E_n \,\sqrt{\det g_{|_{\de\Sigma}}}.
\]
As $Q^{\partial}$ is zero on the $c$ and $E$ coordinates, we get $Q^{\partial\partial}=0$ and $S^{\partial\partial}=0$. Therefore, we get a trivial structure.

We now make a change of coordinates that will make the other examples we want to describe easier to write. Namely, instead of the vector field $\bA$ we consider the corresponding $1$-form $A$, via the metric $g$, and instead of the vector field 
$\bE$ we consider the $(d-1)$-form $B=\iota_{\bE}\sqrt{\det g}$. With these new notations we get
\[
\varpi^{\partial}=\int_\Sigma (\delta B\,\delta A+\delta b\,\delta c),
\]
where we omitted the wedge product symbol from the notation, 
and 
\[
S^{\partial} = \int_\Sigma c\,\dd B.
\]
Note that any reference to the metric $g$ has disappeared. Repeating the above computations, we now get
\[
Q^{\partial}A = \dd c,\quad Q^{\partial}B=0,\quad Q^{\partial}b= \dd B,\quad Q^{\partial}c=0.
\]
If $\Sigma$ has a boundary, we get $\calF_{\de\Sigma}=\{(c,B)\in C^\infty(\de\Sigma)[1]\oplus\Omega^{d-1}(\de\Sigma)\}$ with canonical symplectic structure $\varpi^{\partial\partial}=\int_{\de\Sigma}\delta c\,\delta B$ and with  $Q^{\partial\partial}=0$ and $S^{\partial\partial}=0$. Note that, even though the corner operator $Q^{\partial\partial}$ is trivial, the corner theory still has a nontrivial symplectic space of fields. 

\subsection{Yang--Mills theory}
In the nonabelian case, the fields $A$, $B$, $b$, $c$ are $\frg$-valued,\footnote{For simplicity, we consider YM theory based on a trivial principal bundle over $\Sigma$.} where
 $\frg$ is a Lie algebra endowed with a nondegenerate, invariant inner product $\langle\ ,\ \rangle$. 
 The Gauss law is $\dd_AB=0$, where $\dd_A$ denotes the covariant derivative.
 The BFV symplectic form now reads
 \[
\varpi^{\partial}=\int_\Sigma (\langle\delta B,\,\delta A\rangle+\langle\delta b,\,\delta c\rangle).
\]
As this  notation is a bit heavy, we will omit the inner product $\langle\ ,\ \rangle$ throughout, so we simply write
$\varpi^{\partial}=\int_\Sigma (\delta B\,\delta A+\delta b\,\delta c)$ (one may think of the integral sign to contain the inner product as well, or one may think the inner product to be the Killing form and the integral to incorporate the trace sign). By the same convention, the BFV action reads
\[
S^{\partial} = \int_\Sigma \left( c\,\dd_A B +\frac12 b[c,c]\right),
\]
where the BRST term, linear in $b$, has now appeared. We can also easily calculate
\[
Q^{\partial}A = \dd_A c,\quad Q^{\partial}B=[c,B],\quad Q^{\partial}b=\dd_AB+[c,b],\quad Q^{\partial}c=\frac12 [c,c].
\]

If $\Sigma$ has a boundary, we get $\calF_{\de\Sigma}=\{(c,B)\in (C^\infty(\de\Sigma)[1]\oplus\Omega^{d-1}(\de\Sigma))\otimes\frg\}$ with canonical symplectic structure $\varpi^{\partial\partial}=\int_{\de\Sigma}\delta c\,\delta B$ and with
$Q^{\partial\partial} B=[c,B]$ and $Q^{\partial\partial} c=\frac12 [c,c]$, which is the Hamiltonian vector field of
\[
S^{\partial\partial} = \int_{\partial\Sigma} \frac12 B[c,c].
\]
Now, the BF$^2$V structure is no longer trivial.

If we regard $\calF_{\de\Sigma}$ as $T^*[1](\Omega^{d-1}(\de\Sigma)\otimes\frg)$, we then interpret $S^{\partial\partial}$ as the Poisson bivector field
\[
\pi_2 = \int_{\partial\Sigma} \frac12 B\left[\frac\delta{\delta B},\frac\delta{\delta B}\right].
\]
As this is linear, it can actually be viewed (modulo subtleties due to dualization) as the Poisson structure on $\calG^*$,
where $\calG$ is the Lie algebra $C^\infty(\de\Sigma)\otimes\frg$ with pointwise Lie bracket induced by $\frg$. (We have identified $\frg^*$ with $\frg$ using the inner product and we have regarded $\Omega^{d-1}(\de\Sigma)$ as the dual space of $C^\infty(\de\Sigma)$.) For example, on linear functionals we have
\[
\left\{\int_{\de\Sigma} fB,\int_{\de\Sigma} gB\right\}_2=\int_{\de\Sigma} [f,g]B.
\]

The other natural polarization consists in realizing $\calF_{\de\Sigma}$ as $T^*[1](C^\infty(\de\Sigma)[1]\otimes\frg)$. In this case, we interpret $S^{\partial\partial}$ as the cohomological vector field
\[
\pi_1 = \int_{\partial\Sigma} \frac12 [c,c]\frac\delta{\delta c},
\]
which gives $C^\infty(\de\Sigma)[1]\otimes\frg$ the structure of a $P_\infty$-manifold. With the notations of the previous paragraph, this manifold is the same as $\calG[1]$. Its algebra of functions is the exterior algebra $\Lambda\calG^*$, regarded as a graded commutative algebra, and $\pi_1$ corresponds to the Chevalley--Eilenberg differential.

\begin{remark}
Note that for any $B_0\in \Omega^{d-1}(\de\Sigma)$ we can define a polarization choosing the $B_0$-section of
 $T^*[1](C^\infty(\de\Sigma)[1]\otimes\frg)$ instead of the zero section. In this case, in addition to $\pi_1$ as above, we also get a nontrivial $\pi_0=\int_{\de\Sigma} \frac12 B_0[c,c]$, so we have a curved $P_\infty$ structure.

\end{remark}

\subsection{Chern--Simons theory}
In this case, $\Sigma$ is two-dimensional and the field is a $\frg$-connection one-form $A$, 
where $\frg$ again is a Lie algebra endowed with a nondegenerate, invariant inner product.\footnote{\label{f-triv}For simplicity, we use notations as in the case of a trivial principal bundle. For the general case, the Lie-algebra-valued forms are simply replaced by forms taking value in sections of the adjoint bundle.} The space of fields is endowed with the Atiyah--Bott symplectic form $\varpi^{\partial}_0=\frac12\int_\Sigma \delta A\,\delta A$ and the constraint is that the connection be flat. Therefore, we introduce the BFV structure
\[
\varpi^{\partial}=\int_\Sigma \left(\frac12\delta A\,\delta A+\delta b\,\delta c\right),
\]
\[
S^{\partial} = \int_\Sigma \left(c\,F_A + \frac12 b[c,c]\right),
\]
where $F_A=\dd A+\frac12[A,A]$ is the curvature of $A$.
We now get
\[
Q^{\partial}A = \dd_A c,\quad Q^{\partial}b=F_A+[c,b],\quad Q^{\partial}c=\frac12 [c,c].
\]

If $\Sigma$ has a boundary, we get $\calF_{\de\Sigma}=\{(c,A)\in C^\infty(\de\Sigma)[1]\otimes\frg\oplus\calA(\de\Sigma)\}$, where $\calA$ denotes the space of connection one-forms, with canonical symplectic structure $\varpi^{\partial\partial}=\int_{\de\Sigma}\delta c\,\delta A$ and with
$Q^{\partial\partial} A=\dd_A c$ and $Q^{\partial\partial} c=\frac12 [c,c]$, which is the Hamiltonian vector field of
\[
S^{\partial\partial} = \int_{\partial\Sigma} \frac12 c\,\dd_Ac
=\int_{\partial\Sigma} \left(\frac12 c\,\dd_{A_0}c+\frac12 c[a,c]
\right)
,
\]
where $A_0$ is a reference connection and $a=A-A_0$.\footnote{The introduction of the reference connection $A_0$ has the goal of decoupling the term of the corner action containing a derivative from the one depending on the connection, still maintaining a global description of the theory. In particular, this allows to  identify $\pi_0$ and $\pi_1$ in the second polarization chosen. A different choice of $A_0$ does not change the form of the $P_{\infty}$-structure.}

If we regard $\calF_{\de\Sigma}$ as $T^*[1]\calA(\de\Sigma)$, we then interpret $S^{\partial\partial}$ as the Poisson bivector field
\[
\pi_2 
=\int_{\partial\Sigma} \left(\frac12 \frac\delta{\delta a}\dd_{A_0} \frac\delta{\delta a}
+\frac12 a\left[\frac\delta{\delta a},\frac\delta{\delta a}\right]
\right)
.
\]
In this case, we have an affine Poisson structure which
can be viewed (modulo subtleties due to dualization) as the Poisson structure on $\calG^*$ associated with the central extension of $\calG=C^\infty(\de\Sigma)\otimes\frg$ with pointwise Lie bracket induced by that on $\frg$ by the cocycle
$c(f,g)=\int_{\de\Sigma} f \dd_{A_0} g$. For example, on linear functionals we have
\[
\left\{\int_{\de\Sigma} fa,\int_{\de\Sigma} ga\right\}_2=\int_{\de\Sigma} (f \dd_{A_0} g + [f,g]a).
\]

The other natural polarization consists in realizing $(\calF_{\de\Sigma})_{A_0}$ as $T^*[1](C^\infty(\de\Sigma)[1]\otimes\frg)$. In this case, we interpret $S^{\partial\partial}$ as the inhomogeneous multivector field $\pi=\pi_0+\pi_1$ with 
$\pi_0=\int_{\partial\Sigma} \frac12 c\,\dd_{A_0}c$ and
\[
\pi_1 = \int_{\partial\Sigma} \frac12 [c,c]\frac\delta{\delta c},
\]
which gives $C^\infty(\de\Sigma)[1]\otimes\frg$ the structure of a curved $P_\infty$-manifold. Note that the curving $\pi_0$ is different from zero for every choice of $A_0$.

\begin{remark}
Chern--Simons theory is an example of an AKSZ theory \cite{AKSZ}. In particular, this means that we can write the BF$^n$V structures in a compact way  using superfields. For the cases at hand, we set
$\Tilde A = c + A + b$
in the BFV case and $\Tilde A = c + A$ in the BF$^2$V case. The symplectic forms and actions on $\Sigma$ and on $\partial \Sigma$ now simply read
$\frac12\int_T \delta\Tilde A\delta\Tilde A$ and $\int_T\left(\frac12 \Tilde A\dd\Tilde A+\frac16\Tilde A[\Tilde A,\Tilde A]\right)$, by specializing $T$ to $\Sigma$ or to $\de\Sigma$, respectively.\footnote{Since $\Sigma$ and $\de\Sigma$ have different dimensions the integral will pick different summands of the inhomegeneous field $\Tilde A$, e.g. we will have 
\begin{align*}
  \varpi_{\Sigma}=  \frac12\int_{\Sigma} \delta A\delta A + \delta c\delta b \text{ and } \varpi_{\de \Sigma}= \frac12\int_{\de \Sigma} \delta c\delta A.
\end{align*}}
\end{remark}

\subsection{\texorpdfstring{$BF$}{BF} theory}
In $BF$ theory in $d+1$ dimensions, there are two fields: a $\frg$-connection $A$ and a $\frg$-valued $(d-1)$-form $B$.
Here, $\frg$ is, as before, a Lie algebra endowed with a nondegenerate, invariant inner product.\footnote{See also footnote~\ref{f-triv}.} The symplectic form, for a $d$-manifold $\Sigma$,
is $\varpi^{\partial}_0=\int_\Sigma \delta B\,\delta A$ and the constraints are
\[
\dd_AB=0\qquad\text{and}\qquad F_A+ \Lambda P(B) = 0,
\]
where $\Lambda$ is a constant and $P$ an invariant polynomial of degree $k$ such that $k(d-1)=2$.\footnote{The term $\Lambda P(B)$ is called the cosmological term. If it is absent, one speaks of pure $BF$ theory. In pure $BF$ theory, one does not need the invariant inner product on $\frg$, as one can take $B$ as $\frg^*$-valued.} Note that $P$ may be nontrivial only for $d=2,3$.

For $d=1$, for dimensional reasons the only nontrivial constraint is $\dd_AB=0$, so, in this case, the BFV structure is the same as in the case of Yang--Mills in $1+1$ dimensions.

For $d=2$, $BF$ theory is actually a particular case of Chern--Simons theory with a Lie algebra structure, depending on $\Lambda$, on the vector space $\frg\oplus\frg$. If $\frg=\frso(1,2)$ (or $\frso(3)$) and $B$, viewed as a $3\times3$ tensor field, is nondegenerate, it is $2+1$ (Euclidean) gravity with cosmological constant $\Lambda$ in the coframe formulation.

In the rest of the section, we focus on the case $d=3$, which, for $\frg=\frso(1,3)$ (or $\frso(4)$), is related to   $3+1$ (Euclidean) gravity with cosmological constant $\Lambda$ in the coframe formulation. For definiteness, we write the constraints as
\[
\dd_AB=0\qquad\text{and}\qquad F_A+ \Lambda B = 0.
\]
In the BFV formalism, we then need two kinds of ghosts to implement them. The beginning of the BFV action is
\[
S^{\partial}=\int_\Sigma (c\,\dd_A B+\tau\, (F_A+ \Lambda B))+\cdots,
\]
with $c\in\Omega^0(\Sigma)[1]\otimes\frg$ and $\tau\in\Omega^1(\Sigma)[1]\otimes\frg$. 

Note that the $\tau$-dependent Hamiltonian vector field acts on $A$ as $\Lambda\tau$ and on $B$ as $\dd_A\tau$.
Therefore, if $\tau$ is of the form $\dd_A\phi$ for some $0$-form $\phi$, it acts on $A$ as a gauge transformation.  Moreover, it acts on $B$ as $[F_A,\phi]$. If $F_A+\Lambda B=0$, which is what the constraint imposes, it acts also on $B$ as a gauge transformation. This leads to a redundancy to the  $c$-dependent Hamiltonian vector field. To avoid it, one has to mod out $\tau$ by such transformations. If the momentum for $\tau$ is denoted $B^+$, then we add the term $\int_\Sigma \phi\,\dd_AB^+$ to the BFV action, for its Hamiltonian vector field acts on $\tau$ precisely as $\dd_A\phi$.
Note that $\phi$ is now considered as a new ghost (actually a ghost-for-ghost), which is assigned even parity and degree equal to two. It also comes with its own momentum.

As $BF$ theory is an AKSZ theory, we will use the notation standard in that context. Namely, we group the fields into superfields,
\begin{align*}
\Tilde A &= c+A+B^++\tau^+,\\
\Tilde B &= \phi + \tau + B + A^+,
\end{align*}
where the fields appear in decreasing order w.r.t.\ degree and in increasing order w.r.t.\ form degree. The BFV symplectic form is
\[
\varpi^{\partial}=\int_\Sigma\delta\Tilde B\,\delta\Tilde A
=\int_\Sigma(\delta A^+\,\delta c + \delta B\,\delta A + \delta\tau\,\delta B^+ +\delta\phi\,\delta\tau^+),
\]
from which it is clear that our notation for the momenta of $c$, $\tau$, and $\phi$ are $A^+$, $B^+$, and $\tau^+$, respectively. The BFV action reads
\begin{align*}
S^{\partial}&=\int_\Sigma\left(\Tilde BF_{\Tilde A}+\frac12\Lambda\Tilde B\Tilde B\right)\\
&=\int_\Sigma\left(\frac12A^+[c,c]+B\,\dd_Ac+\tau\,(F_A+[c,B^+])+\phi\,(\dd_AB^++[c,\tau^+])+\Lambda\,(B\tau+A^+\phi)
\right),
\end{align*}
from which we get
\begin{align*}
    Q^{\partial} c&=\frac12[c,c]+\Lambda\phi,&\quad Q^{\partial} A&=\dd_Ac+\Lambda\tau,\\
Q^{\partial} B^+&=F_A+\Lambda B+ [c,B^+],&\quad Q^{\partial} \tau^+&=\dd_AB^++[c,\tau^+]+\Lambda A^+,
\end{align*}
and
\begin{align*}
    Q^{\partial}\phi&=[c,\phi], & Q^{\partial}\tau&=\dd_A\phi+[c,\tau],\\ Q^{\partial} B&=\dd_A\tau+[c,B]+[\phi,B^+],& Q^{\partial} A^+&=\dd_AB+[c,A^+]+[B^+,\tau]+[\tau^+,\phi].
\end{align*}

If $\Sigma$ has a boundary, we get that the coordinates of $\calF_{\de\Sigma}$ can also be grouped into superfields
\begin{align*}
\Tilde A &= c+A+B^+,\\
\Tilde B &= \phi + \tau + B.
\end{align*}
The BF$^2$V symplectic form turns out to be
\[
\varpi^{\partial\partial}=\int_{\partial\Sigma}\delta\Tilde B\,\delta\Tilde A
=\int_{\partial\Sigma}(\delta B\,\delta c + \delta \tau\,\delta A + \delta\phi\,\delta B^+ ).
\]
{}From
\begin{align*}
Q^{\partial\partial} c&=\frac12[c,c]+\Lambda\phi,& Q^{\partial\partial} A&=\dd_Ac+\Lambda\tau,& Q^{\partial\partial} B^+&=F_A+\Lambda B+ [c,B^+],\\
Q^{\partial\partial} \phi&=[c,\phi],& Q^{\partial\partial} \tau&=\dd_A\phi+[c,\tau],& Q^{\partial\partial} B&=\dd_A\tau+[c,B]+[\phi,B^+],\\
\end{align*}
we get the BF$^2$V action
\begin{align*}
&S^{\partial\partial}=\int_{\partial\Sigma}\left(\Tilde BF_{\Tilde A}+\frac12\Lambda\Tilde B\Tilde B\right)\\
&=\int_{\partial\Sigma}\left(\frac12B[c,c]+\tau\,\dd_Ac+\phi\,(F_A+[c,B^+])+\Lambda\left(\frac12\tau\tau+B\phi\right)
\right)\\
&=\int_{\partial\Sigma}\left(\frac12B[c,c]+\tau\,(\dd_{A_0}c+[a,c])+\phi\,\left(F_{A_0}+\dd_{A_0}a+\frac12[a,a]+[c,B^+]\right)+\Lambda\left(\frac12\tau\tau+B\phi\right)
\right)
\end{align*}
where $A_0$ is a reference connection and $a=A-A_0$.

\begin{remark}
Note that even    in the abelian case the corner action is nontrivial.
\end{remark}

One natural polarization
 consists in realizing $\calF_{\de\Sigma}$ as the shifted cotangent bundle of the space
$\calN$ with coordinates $A$, $B$, and $B^+$, by choosing $\{c=\phi=\tau=0\}$ as the reference Lagrangian submanifold. This corresponds to having $\pi=\pi_1+\pi_2$ with
\begin{align*}
\pi_1 &= \int_{\de\Sigma} (F_A+\Lambda B)\frac\delta{\delta B^+},\\
\pi_2 &= \int_{\de\Sigma}\left(\frac12B\left[\frac\delta{\delta B},\frac\delta{\delta B}\right]
+\frac\delta{\delta a}\dd_{A_0}\frac\delta{\delta B} + a\left[\frac\delta{\delta a},\frac\delta{\delta B}\right]
+B^+\left[\frac\delta{\delta B^+},\frac\delta{\delta B}\right] + \frac12\Lambda\frac\delta{\delta a}\frac\delta{\delta a}
\right).
\end{align*}
In other words, we split functions on $\calF_{\de\Sigma}$ as $\frp\oplus\frh$ with $\frp$ the subalgebra of functions of positive degree and $\frh$ the subalgebra of functions of nonpositive degree, and the construction turns $\frh$ into a differential graded Poisson algebra. The degree zero part $\frh_0$, consisting of functions on $\calA(\de\Sigma)\oplus\Omega^2(\de\Sigma)\otimes\frg\ni(A,B)$, is a Poisson subalgebra. Actually, we may view the affine Poisson structure on $\calA(\de\Sigma)\oplus\Omega^2(\de\Sigma)\otimes\frg=(A_0+\Omega^1(\de\Sigma)\otimes\frg)\oplus\Omega^2(\de\Sigma)$ as the one on the dual $\calG^*$ associated with the central extension of $\calG=(\Omega^1(\de\Sigma)\oplus\Omega^0(\de\Sigma))\otimes\frg$ with pointwise Lie bracket induced by that on the semidirect sum $\frg\rtimes_{\operatorname{ad}}\frg$ by the cocycle 
$c(\alpha\oplus f,\beta\oplus g)=\int_{\de\Sigma}\left(\alpha\dd_{A_0}g- \beta\dd_{A_0} f+\Lambda\alpha\beta\right)$. For example, on linear functionals we have
\begin{align*}
\left\{\int_{\de\Sigma} \alpha a,\int_{\de\Sigma} \beta a\right\}_2&=\Lambda\int_{\de\Sigma} \alpha\beta,\\
\left\{\int_{\de\Sigma} \alpha a,\int_{\de\Sigma} fB\right\}_2&=\int_{\de\Sigma} (\alpha \dd_{A_0} f + [\alpha,f]a),\\
\left\{\int_{\de\Sigma} fB,\int_{\de\Sigma} gB\right\}_2&=\int_{\de\Sigma} [f,g]B.
\end{align*}
The degree-zero $\pi_1$-cohomology is the quotient of $\frh_0$ by the ideal generated by $F_A+\Lambda B$. Geometrically, this corresponds to restricting the above Poisson structure to the Poisson submanifold
$\{(A,B)\ |\ F_A+\Lambda B=0\}$.

Another natural polarization consists in viewing $\calF_{\de\Sigma}$ as the shifted cotangent bundle of the space $\Tilde\calA$ with coordinates $c$, $A$, and $B^+$, by choosing $\{\Tilde B=0\}$ as the reference Lagrangian submanifold. This corresponds to having $\pi=\pi_1+\pi_2$ with
\begin{align*}
\pi_1 &= \int_{\de\Sigma}\left(\frac12[c,c]\frac\delta{\delta c} + \dd_Ac\frac\delta{\delta A}+(F_A+[c,B^+])\frac\delta{\delta B^+}\right),\\
\pi_2 &=\Lambda\int_{\de\Sigma}\left(\frac12\frac\delta{\delta A}\frac\delta{\delta A}+\frac\delta{\delta c} \frac\delta{\delta B^+}\right).
\end{align*}
In particular, on $C^\infty(\Tilde\calA)$ we have a differential defined by
\[
\pi_1c=\frac12[c,c],\quad \pi_1A=\dd_Ac,\quad \pi_1B^+=F_A+ [c,B^+].
\]
If $\Lambda\not=0$, we also have a constant, nondegenerate Poisson bracket.

One last interesting polarization, which turns out to be important for the rest of this paper, consists instead in viewing $\calF_{\de\Sigma}$ as the shifted cotangent bundle of the space $\Tilde\calB$ with coordinates $\phi$, $\tau$, and $B$, by choosing $\{\Tilde A=A_0\}$ as the reference Lagrangian submanifold. In this case, we have $\pi=\pi_0+\pi_1+\pi_2$ with
\begin{align*}
\pi_0 &= \int_{\de\Sigma}\left(\phi F_{A_0} +\Lambda\,\left(\frac12\tau\tau+B\phi
\right)
\right),\\
\pi_1 &= \int_{\de\Sigma}\left(\dd_{A_0}\tau\frac\delta{\delta B}+\dd_{A_0}\phi\frac\delta{\delta \tau}
\right),\\
\pi_2 &= \int_{\de\Sigma}\left(\frac12B\left[\frac\delta{\delta B},\frac\delta{\delta B}\right]
+\tau\left[\frac\delta{\delta \tau},\frac\delta{\delta B}\right]
+\frac12\phi\left[\frac\delta{\delta \tau},\frac\delta{\delta \tau}\right]
+\phi\left[\frac\delta{\delta \phi},\frac\delta{\delta B}\right]
\right).
\end{align*}
This makes $C^\infty(\Tilde\calB)$ into a curved $P_\infty$ algebra. If $\Lambda=0$, it can be made flat by choosing the reference connection $A_0$ to be flat. It is useful, for further reference, to observe that there is a $P_\infty$ subalgebra generated by the following linear local observables:
\begin{align}\label{e:UVW_BF}
    J_\alpha = \int_{\de\Sigma}\alpha B,\quad
    M_\beta = \int_{\de\Sigma}\beta \tau,\quad
    K_\gamma = \int_{\de\Sigma}\gamma \phi,
\end{align}
where $\alpha$, $\beta$, $\gamma$ are $\frg$-valued $0$-, $1$-, and $2-$forms, respectively. We have
\begin{gather*}
\{\}_0 = \int_{\de\Sigma}\left(\phi F_{A_0} +\Lambda\,\left(\frac12\tau\tau+B\phi
\right)\right)\\
\{J_\alpha\}_1=M_{\dd_{A_0}\alpha},\quad
\{M_\beta\}_1=K_{\dd_{A_0}\beta},\quad
\{K_\gamma\}_1=0,\\
\{J_\alpha,J_{\Tilde\alpha}\}_2=J_{[\alpha,\Tilde\alpha]},\quad
\{J_\alpha,M_\beta\}_2=M_{[\alpha,\beta]},\quad
\{J_\alpha,K_\gamma\}_2=K_{[\alpha,\gamma]},\\
\{M_\beta,M_{\Tilde\beta}\}_2=K_{[\beta,\Tilde\beta]},\quad
\{M_\beta,K_\gamma\}_2=0,\quad
\{K_\gamma,K_{\Tilde\gamma}\}_2=0.
\end{gather*}
Also note that $\{\{\}_0\}_1=0$, that $\{M_\beta,\{\}_0\}_2=0=\{K_\gamma,\{\}_0\}_2$, that $\{\{M_\beta\}_1\}_1=0=\{\{K_\gamma\}_1\}_1$, and that $\{\{J_\alpha\}_1\}_1=\{J_\alpha,\{\}_0\}_2$. Observe that for $\Lambda=0$ we can also write $\{\{J_\alpha\}_1\}_1=K_{[F_{A_0},\alpha]}$.
It is also instructive to compute the above expressions using the derived brackets corresponding to the splitting with $\frh=C^\infty(\calB)$ and $\frp$ the ideal in $C^\infty(\calF_{\de\Sigma})$ generated by $C^\infty(\calA-A_0)$.
In this case, the projection $P\colon C^\infty(\calF_{\de\Sigma})\to C^\infty(\calB)$ simply consists in setting $A$ equal to $A_0$ and $c$ and $B^+$ to zero.
We see that $\{\}_0=PS^{\partial\partial}$. We can also, e.g., compute
\[
\{J_\alpha\}_1=PQ^{\partial\partial} J_\alpha=P\int_{\de\Sigma}\alpha(\dd_A\tau+[c,B]+[\phi,B^+])=
\int_{\de\Sigma}\alpha\dd_{A_0}\tau=M_{\dd_{A_0}\alpha}.
\]
Similarly, we get
\[
\{J_\alpha,M_\beta\}_2=P\{J_\alpha,Q^{\partial\partial} M_\beta\}
=
P\left\{\int_{\de\Sigma}\alpha B,\int_{\de\Sigma}\beta (\dd_A\phi+[c,\tau])\right\}
=
P\int_{\de\Sigma}[\alpha,\beta]\tau=M_{[\alpha,\beta]}.
\]

Note that, when $\Lambda=0$, the above algebra closes also under the nullary operation, since we can write
\[
\{\}_0 = K_{F_{A_0}}.
\]
Otherwise, we have to add more generators. First of all, we introduce
\[
C_\mu= \int_{\de\Sigma}\mu\,\left(\frac12\tau\tau+B\phi
\right),
\]
where $\mu$ is a function,
so that we have
\[
\{\}_0 = K_{F_{A_0}}+C_\Lambda,
\]
where we view $\Lambda$ as a constant function. The algebra now closes as long as $C_\mu$ is defined for constant functions $\mu$ only. 

It is however possible, and natural, to extend the algebra allowing for arbitrary functions $\mu$. In this case, we have to introduce
\begin{align*}
    D_\nu &= \int_{\de\Sigma}\nu\tau\phi,\\
    E_\rho &= \frac12\int_{\de\Sigma}\rho \phi^2.
\end{align*}
It can be readily verified that the 2-brackets of $C$, $D$, and $E$ among themselves or with $J$, $M$, and $K$ all vanish. As for the unary operations, we have
\[
\{C_\mu\}_1=D_{\dd\mu},\qquad\{D_\nu\}_1=E_{\dd\nu},\qquad \{E_\rho\}_1=0.
\]

\section{Boundary structure and BFV data for Palatini--Cartan theory}\label{s:BFVPC}

    The starting point for the construction of the BF$^2$V structure is the BFV boundary structure. In the Palatini--Cartan formalism, this is described in \cite{CCS2020}.
    
We recall here the relevant quantities of this construction.
We consider a four-dimensional closed, oriented\footnote{Orientability is not really necessary, see \cite{CCS2020}, but we assume it here for simplicity. We also assume compactness to avoid discussing vanishing conditions on the fields; see also footnote~\ref{f:compact} on page~\pageref{f:compact}. In the second part of the discussion, $M$ will be allowed to have a boundary $\Sigma$, which later will also be allowed to have a boundary, so $M$ will actually be a manifold with corners.}  smooth manifold $M$ together with a reference Lorentzian structure
so that we can reduce the frame bundle to an $SO(3,1)$-principal bundle $P \rightarrow M$. 
We denote by $\mathcal{V}$ the associated vector bundle by the standard representation. Each fibre of $\mathcal{V}$ is isomorphic to a four-dimensional vector space $V$ with a Lorentzian inner product $\eta$ on it. The inner product allows the identification $\mathfrak{so}(3,1) \cong \bigwedge^2 {V}$. Let now $\Sigma=\partial M$ be the boundary of $M$ and denote with $\mathcal{V}_\Sigma$ the restriction $\mathcal{V}|_{\Sigma}$. We define the following shorthand notation:
\begin{align*}
     \Omega_{\partial}^{i,j}:= \Omega^i\left(\Sigma, \textstyle{\bigwedge^j} \mathcal{V}_{\Sigma}\right).
\end{align*}
\begin{remark}
Throughout the article, we will refer to the local dimensions of the spaces $\Omega^{i,j}$ as the number of degrees of freedom of the space.
Note that this dimension is also the same as their rank as $C^\infty$ modules. 
\end{remark}
On $\Omega_{\partial}^{i,j}$, we also define the following maps
\begin{align*}
    W_{\partial}^{(i,j)}\colon \Omega_{\partial}^{i,j}  & \longrightarrow  \Omega_{\partial}^{i,j}\\
    X  & \longmapsto  X \wedge e|_{\Sigma}.
\end{align*}
Usually, we will omit writing the restriction of $e$ to the manifold $\Sigma$. The properties of these maps are collected in Appendix \ref{s:appendix_notation}.

{

 We assume $\mathcal{V}_\Sigma$ to be isomorphic to $T\Sigma\oplus\underline{\bbR}$, as is the case if we think of it as the restriction to the boundary of a vector bundle isomorphic to the tangent bundle of the bulk, and we take a nowhere vanishing section $\epsilon_n$ of the summand $\underline{\bbR}$. We then define the space $\Omega_{\epsilon_n}^1(\Sigma, \mathcal{V}_\Sigma)$ 
 to consist of bundle maps $e\colon T\Sigma\to\mathcal{V}_\Sigma$ such that the three components of $e$ together with $\epsilon_n$ form a basis. Equivalently, we may require $eee\epsilon_n$ to be different from zero everywhere.\footnote{As already noted in \cite{CCS2020}, the results are independent on the choice of $\epsilon_n$. In particular, this is clear if $\Sigma$ is spacelike, since the space $\Omega_{\epsilon_n}^1(\Sigma, \mathcal{V}_\Sigma)$ of space-like vectors does not depend on the choice of a specific time-like $\epsilon_n$. Note that in \cite{CCS2020} the space here denoted by $\Omega_{\epsilon_n}^1(\Sigma, \mathcal{V}_\Sigma)$ was denoted by  $\Omega_\text{nd}^1(\Sigma, \mathcal{V}_\Sigma)$.}

As a consequence of this, the field $e$ together with $\epsilon_n$ defines an isomorphism $T\Sigma\oplus\underline{\bbR}\to\mathcal{V}_\Sigma$. Denoting by $f\colon \mathcal{V}_\Sigma\to T\Sigma\oplus\underline{\bbR} $ its inverse and by $\pi_{T\Sigma}$ the projection $T\Sigma\oplus\underline{\bbR}\to T\Sigma$, we have a map
\begin{align}\label{e:hatmap}
\widehat{ \bullet}\ \colon
\begin{array}[t]{ccc}
\Gamma(\mathcal{V}_\Sigma) & \to &\mathfrak X(\Sigma) \\
\nu &\mapsto & \widehat\nu:=\pi_{T\Sigma}(f(\nu))
\end{array}
\end{align}
Note that the definition of the hat map really depends on the choice of $\epsilon_n$ and the field $e$, even though we hide it in the notation.

In local coordinates, the hat map has the following description. We denote by $e_a$, $a=1,2,3$, the three components of the $\mathcal{V}_\Sigma$-valued one-form $e$. Then, for a given $\nu\in\Gamma(\mathcal{V}_\Sigma)$, there are uniquely determined functions $\nu^{(a)}$, $a=1,2,3$, and $\nu^{(n)}$ such that

\begin{align*}
    \nu= \nu^{(a)}e_a + \nu^{(n)}\epsilon_n.
\end{align*}
The induced hat vector field is then 
\[
\widehat{\nu}= \nu^{(a)}\frac\partial{\partial x^a}. 
\]

We also consider the space 
\begin{align*}
    T^* \left(\Omega_{\partial}^{0,2}[1]\oplus \mathfrak{X}[1](\Sigma) \oplus C^\infty[1](\Sigma)\right)
\end{align*}
where the corresponding fields are denoted by $c \in\Omega_{\partial}^{0,2}[1]$, $\xi \in\mathfrak{X}[1](\Sigma)$,  $\lambda\in \Omega^{0,0}[1]$,
$\gamma^\dag\in\Omega_{\partial}^{3,2}[-1]$, and $y^\dag\in\Omega_{\partial}^{3,3}[-1]$.\footnote{Note that here we are using an isomorphism defined by $e$ in order to identify the fiber of $ T^* \left(\mathfrak{X}[1](\Sigma) \oplus C^\infty[1](\Sigma)\right)$ with $\Omega_{\partial}^{3,3}[-1]$.}
The space of boundary fields is the bundle 
\begin{equation*}
\mathcal{F}^{\partial} \longrightarrow \Omega_{\epsilon_n}^1(\Sigma, \mathcal{V}_\Sigma)\oplus T^* \left(\Omega_{\partial}^{0,2}[1]\oplus \mathfrak{X}[1](\Sigma) \oplus C^\infty[1](\Sigma)\right),
\end{equation*}
with local trivialisation on an open $\mathcal{U}_{\Sigma} \subset \Omega_{\epsilon_n}^1(\Sigma, \mathcal{V}_\Sigma)\oplus T^* \left(\Omega_{\partial}^{0,2}[1]\oplus \mathfrak{X}[1](\Sigma) \oplus C^\infty[1](\Sigma)\right)$ given by
\begin{equation*}
\mathcal{F}^{\partial}\simeq \mathcal{U}_{\Sigma} \times \mathcal{A}^\text{red}(\Sigma),
\end{equation*}
where $\mathcal{A}^\text{red}(\Sigma)$ is the space of connections $\omega$ (on $P|_{\Sigma}$)  such that 
\begin{equation}\label{e:structural_constraint}
\epsilon_n \dd_{\omega}e + \iota_{\widehat{X}}\gamma^{\dag}= e \sigma 
\end{equation} 
for some $\sigma \in \Omega_{\partial}^{1,1}$ and $X=[c, \epsilon_n]+ L_{\xi}^{\omega} \epsilon_n$.
}
The constraint \eqref{e:structural_constraint} is called structural constraint. 
The BFV action and symplectic form are respectively:
\begin{align}
S^{\partial}&= \int_{\Sigma} 
\Big(c e \dd_{\omega} e + \iota_{\xi} e e F_{\omega} +\lambda \epsilon_n e F_{\omega} + \frac{1}{3!}\lambda \epsilon_n \Lambda e^3+\frac{1}{2} [c,c] \gamma^{\dag} - L^{\omega}_{\xi} c \gamma^{\dag}+ \frac{1}{2} \iota_{\xi}\iota_{\xi} F_{\omega}\gamma^{\dag}\\ &
 \phantom{=}+[c, \lambda \epsilon_n ]y^{\dag} - L^{\omega}_{\xi} (\lambda \epsilon_n)y^{\dag} - \frac{1}{2}\iota_{[\xi,\xi]}e y^{\dag}\Big),\nonumber\\
\varpi^{\partial} &= \int_{\Sigma} \left(e \delta e \delta \omega + \delta c\delta \gamma^{\dag} - \delta \omega \delta (\iota_\xi \gamma^{\dag}) + \delta \lambda \epsilon_n \delta y^\dag+\iota_{\delta \xi} \delta (e y^\dag)\right).
\end{align}

\begin{remark}
For simplicity, we consider in this paper only the case of dimension $N=4$. However, some of the considerations of this article can be extended to the higher-dimensional cases. This can be done in the same way in which we can extend the boundary results on the boundary from the case $N=4$ to a generic $N\geq 4$ (see \cite{CCS2020}).
Furthermore, in this and the following sections, we assume that the cosmological constant vanishes: $\Lambda=0$. In Section~\ref{s:cosmological} we will discuss the small corrections that have to be implemented when the cosmological constant is nonzero.
\end{remark}

The boundary structure is completed by the cohomological vector field $Q^{\partial}$ defined as the Hamiltonian vector field of $S^{\partial}$ with $\partial \Sigma = \emptyset $. Its expression (in components) reads:
\begin{subequations}\label{e:Q-boundary}
\begin{align}
Q^{\partial} e &= [c,e] + L_\xi^\omega e + \dd_{\omega}(\lambda \epsilon_n) + \lambda \sigma, \\
Q^{\partial} \omega &= \dd_\omega c - \iota_\xi F_\omega + \lambda (W_{\partial}^{(1,2)})^{-1}(\epsilon_n F_{\omega}+\iota_{\widehat{X}}y^\dag)+ \frac{1}{2}\lambda \epsilon_n \Lambda e + \mathbb{V}_{\omega}, \\
Q^{\partial} c &= \frac{1}{2}[c,c] + \frac{1}{2}\iota_\xi\iota_\xi F_\omega+ \lambda \iota_\xi (W_{\partial}^{(1,2)})^{-1}(\epsilon_n F_{\omega}+X^{(a)}y_a^\dag)+ \iota_{\xi}\mathbb{V}_{\omega},\\
Q^{\partial} \lambda &= [c, \lambda \epsilon_n ]^{(n)} + (L_\xi^\omega \lambda \epsilon_n)^{(n)}, \\
Q^{\partial} \xi&= \lambda \widehat{X} + \frac{1}{2}[\xi, \xi],\\
Q^{\partial} \gamma^{\dag} &=  e d_\omega e +[c,\gamma^{\dag}] + L_\xi^\omega \gamma^{\dag} + [\lambda \epsilon_n, y^\dag],\\
e_a Q^{\partial} y^\dag &= e_a[c,y^\dag] + e_a L_\xi^\omega y^\dag + e_a e F_\omega +  (\gamma_a^{\dag} \lambda (W_{\partial}^{(1,2)})^{-1}(\epsilon_n F_{\omega}+\iota_{\widehat{X}}y^\dag)\\ & \qquad+ \lambda \sigma_a y^\dag + \mathbb{V}_{\omega}\gamma_a^{\dag},  \nonumber\\
\epsilon_n Q^{\partial} y^\dag &= \epsilon_n[c,y^\dag] + \epsilon_n L_\xi^\omega y^\dag + \epsilon_n e F_\omega + \frac{1}{3!}\Lambda \epsilon_n e^3,
\end{align}
\end{subequations}
where $X=[c,  \epsilon_n ] + L_\xi^\omega (\epsilon_n)$ and $e \mathbb{V}_{\omega}=0$.
\begin{remark}
    The map $W_\partial^{(1,2)}$ is surjective but not injective (see Appendix \ref{s:appendix_notation} for more details), so we can choose a preimage defined up to terms in the kernel of $W_\partial^{(1,2)}$, denoted here by $\mathbb{V}_{\omega}$. This is 
    fixed by requiring that the action of the vector field $Q^{\partial}$ preserves the structural constraint \eqref{e:structural_constraint}, for some choice of the action of $Q^{\partial}$ on $\sigma$; i.e., we require (\cite{CCS2020}) that 
    \begin{align*}
        Q^{\partial}(\epsilon_n d_{\omega}e + \iota_{\widehat{X}}\gamma^{\dag})= Q^{\partial}e \sigma  + e Q^{\partial} \sigma.
    \end{align*}
    This way we get  
    an inverse $(W_\partial^{(1,2)})^{-1}$. 
    Comparing this expression with the corresponding one of the three-dimensional theory \cite{CaSc2019}, we also note that the terms containing the inverse of the function $W_{\partial}^{(1,2)}$ and the auxiliary field $\sigma$ constitute exactly the difference between the two expressions.
\end{remark}

\section{Corner structure of Palatini--Cartan formalism}
\label{s:BFFV_Corner}

\subsection{Corner induced structure}\label{s:cornerind}
From a boundary BFV action, we can now induce a corner structure following the procedure recalled in Section~\ref{s:relind}. From now on, we assume that the manifold $\Sigma$ has a nonempty boundary $\partial \Sigma= \Gamma$.\footnote{Later, we can drop the hypothesis of $\Gamma$ being a boundary and we can just consider the structures to be defined on a generic two-dimensional manifold $\Gamma$.}
In this and in the following sections, we describe the relaxed BF$^2$V structure on the corner. In particular, we
have the following result: \begin{proposition}\label{prop:corner_extension_PC}
The BFV theory $\mathfrak{F}^{\filt{1}}_{PC}=(\mathcal{F}^{\partial}_{PC}, S^{\partial}_{PC}, \varpi^{\partial}_{PC}, Q^{\partial}_{PC})$ is not 1-extendable.
\end{proposition}
We will then construct some associated $P_\infty$ algebras and will highlight a relation with $BF$ theory (Section \ref{s:Pstru}). We will also describe particular cases where we freeze some of the fields or do some partial reductions (Section \ref{s:simplified}).

\begin{remark}
Note that the four-dimensional case differs from the three-dimensional case. In this last, it has been proven in \cite{CaSc2019} that it is possible to extend the BFV theory to a BF$^2$V theory on the corner. The differencesbetween the three- and four-dimensional cases are to be accounted to mathematical properties of this particular formulation of the BV theories. At the moment neither do we have a physical interpretation of Proposition \ref{prop:corner_extension_PC}, whose content might even just be a mathematical artifact, nor do we have 
a meaningful statement about such an interpretation; hence, we postpone these to future work.
\end{remark}

Before proving Proposition~\ref{prop:corner_extension_PC}, let us introduce some further piece of notation, similarly to what we have done for the boundary structure. Let $M$ be a smooth manifold of dimension $4$ with corners and let us denote by $\Sigma= \partial M$ its three-dimensional boundary and by $\Gamma= \partial\partial M$ its $2$-dimensional corner. Furthermore, we will use the  notation $\mathcal{V}_{\Gamma}$  for the restriction of $\mathcal{V}_\Sigma$ to $\Gamma$. We  define 
\begin{align*}
 \Omega_{\partial\partial}^{i,j}:= \Omega^i\left(\Gamma, \textstyle{\bigwedge^j} \mathcal{V}_{\Gamma}\right).
\end{align*}

On $\Omega_{\partial\partial}^{i,j}$, we define the following map:
\begin{align*}
    W_{\partial\partial}^{  (i,j)}\colon \Omega_{\partial\partial}^{i,j}  & \longrightarrow  \Omega_{\partial\partial}^{i,j}\\ 
    X  & \longmapsto  X \wedge e|_{\Gamma} .
\end{align*}
\begin{remark}
As before,  we will omit writing the restriction of $e$ to the corner $\Gamma$.
\end{remark}

The properties of these maps are collected in  Appendix \ref{s:appendix_notation}. Furthermore, we recall that the restriction to $\Gamma$ of a vector field $\nu\in\mathfrak{X}(\Sigma)$ contracted through the interior product with a one form $\beta \in \Omega^1(\Sigma)$ reads
\begin{align*}
    \iota_{\nu}\beta= \iota_{\nu|_{\Gamma}}\beta|_{\Gamma}+ \beta_m \nu^m,
\end{align*}
where the index $m$ denotes the components transversal to the corner.
For simplicity, we will omit the restrictions to $\Gamma$.
\begin{proof}[Proof of Proposition~\ref{prop:corner_extension_PC}]
From the variation of the boundary action, using the formula
\begin{align*}
    \delta S^{\partial} = \iota_{Q^{\partial}} \varpi^{\partial} + \widetilde{\alpha}^{\partial},
\end{align*}
we can deduce the pre-corner (or pre-codimension-2) one form
\begin{align*}
\widetilde{\alpha}^{\partial}&= \int_{\Gamma} ( c e \delta e - \iota_{\xi} e e \delta \omega - e_m \xi^m e \delta \omega-\lambda \epsilon_n e  \delta \omega - \delta c  \gamma_m^{\dag}\xi^m-  \delta \omega\iota_{\xi}\gamma_m^{\dag}\xi^m\nonumber\\ &\phantom{=}
- \delta (\lambda \epsilon_n) \iota_{\xi}y^{\dag}- \delta (\lambda \epsilon_n) y_m^{\dag} \xi^m - \iota_{\delta \xi} e y_m^{\dag} \xi^m  + e_m \delta \xi^m y_m^{\dag} \xi^m ).
\end{align*} 

Taking its variation, we obtain the pre-corner two-form:
\begin{align}\label{e:precorner_two_form}
\widetilde{\varpi}^{\partial}= \delta \widetilde{\alpha}^{\partial} &= \int_{\Gamma} ( \delta c e \delta e - \iota_{\delta \xi} e e \delta \omega- \iota_{\xi}( e \delta e) \delta \omega -\delta e_m \xi^m e \delta \omega+ e_m \delta\xi^m e \delta \omega- e_m \xi^m \delta e \delta \omega\\ 
&\phantom{=} -\delta \lambda \epsilon_n e  \delta \omega -\lambda \epsilon_n \delta e  \delta \omega- \delta c  \gamma_m^{\dag}\delta \xi^m- \delta c  \delta \gamma_m^{\dag}\xi^m-  \delta \omega\delta(\iota_{\xi}\gamma_m^{\dag}\xi^m)\nonumber\\ &\phantom{=}
+ \delta (\lambda \epsilon_n) \delta y_m^{\dag} \xi^m+ \delta (\lambda \epsilon_n) y_m^{\dag} \delta\xi^m 
+ \iota_{\delta \xi} \delta e y_m^{\dag} \xi^m +\iota_{\delta \xi} e \delta y_m^{\dag} \xi^m- \iota_{\delta \xi} e y_m^{\dag} \delta \xi^m \nonumber\\ &\phantom{=}
  + \delta e_m \delta \xi^m y_m^{\dag} \xi^m- e_m \delta \xi^m \delta y_m^{\dag} \xi^m  + e_m \delta \xi^m y_m^{\dag} \delta \xi^m).  \nonumber
\end{align} 
In order to proceed, we have to check if this two-form is pre-symplectic, i.e., if the kernel of the corresponding map
\begin{align*}
\widetilde{\varpi}^{\partial \sharp}: T \widetilde{\mathcal{F}}^{\partial}& \rightarrow T^* \widetilde{\mathcal{F}}^{\partial}\\
X & \mapsto \widetilde{\varpi}^{\partial \sharp}(X)=\widetilde{\varpi}^{\partial}(X, \cdot)
\end{align*} is regular.
The equations defining the kernel are:
\begin{subequations}
\begin{align}
\delta c: &\quad  e X_e  +X_{\gamma_m^{\dag}} \xi^m- \gamma_m^{\dag} X_{\xi^m}=0 \label{e:kernel_Xe},\\
\delta e: &\quad  e X_c  - e \iota_\xi X_{\omega} - \lambda \epsilon_n X_{\omega}- e_m \xi^m  X_{\omega} -\iota_{X_{\xi}}y_m^\dag \xi^m=0,\label{e:kernel_Xc}\\
\delta \xi: &\quad e_{\bullet} e X_{\omega}-X_{\omega}c_{m\bullet}^{\dag} \xi^m+ (X_e)_\bullet y_m^{\dag} \xi^m + e_{\bullet}X_{y_m^{\dag}} \xi^m- e_{\bullet}y_m^{\dag}X_{\xi^m}=0,\\
\delta \omega: &\quad - \iota_{X_\xi} e e - \iota_{\xi}( e  X_e) -X_{e_m} \xi^m e + e_m X_{\xi^m} e - e_m \xi^m X_e \nonumber\\ & \qquad -X_\lambda \epsilon_n e   -\lambda \epsilon_n X_e  -  X_{(\iota_{\xi}\gamma_m^{\dag}\xi^m)}=0,\\
\delta e_m: & \quad - \xi^m e X_\omega+X_{\xi^m} y_m^{\dag} \xi^m=0,\\
\delta \xi^m: &\quad e_m e X_\omega- X_c  \gamma_m^{\dag}-  X_\omega \iota_{\xi}\gamma_m^{\dag}+  X_\lambda \epsilon_n y_m^{\dag} - \iota_{X_\xi} e y_m^{\dag}
  + X_{e_m} y_m^{\dag} \xi^m \nonumber\\ 
  & \qquad - e_m X_{ y_m^{\dag}} \xi^m  + 2 e_m y_m^{\dag} X_{\xi^m}=0,\\
\delta \lambda: &\quad -\epsilon_n e  X_\omega+ \epsilon_n X_{y_m^{\dag}} \xi^m+ \epsilon_n y_m^{\dag} X_{\delta\xi^m}=0,\\
\delta \gamma_m^{\dag}: &\quad - X_c  \xi^m+ \iota_{\xi}X_\omega \xi^m=0,\\
\delta y_m^{\dag}: &\quad +X_\lambda \epsilon_n \xi^m +\iota_{X_\xi} e \xi^m - e_m X_{\xi^m} \xi^m  =0.
\end{align}
\end{subequations}
Let us consider \eqref{e:kernel_Xe} and \eqref{e:kernel_Xc}. They can be solved only if the functions $W_{\partial\partial }^{(1,1)}$ and $W_{\partial\partial}^{(0,2)}$ are invertible. 
However, from Lemma \ref{lem:Wmaps_properties_corner} in  Appendix~\ref{s:appendix_notation} we gather that both $W_{\partial\partial }^{(1,1)}$ and $W_{\partial\partial}^{(0,2)}$ are neither injective nor surjective. In particular, $\dim \Ima W_{\partial\partial }^{(1,1)}= \dim \Ima W_{\partial\partial}^{(0,2)}=5$, while the respective codomains $\Omega_{\partial\partial}^{1,1}$ and $\Omega_{\partial\partial}^{0,2}$ have dimension $6$ and $8$, respectively. Hence we deduce that these two equations are singular and so is the kernel of $\widetilde{\varpi}^{\partial \sharp}$. 

Therefore, it is not possible to perform a symplectic reduction, and the BFV data do not induce a 1-extended BFV theory.
\end{proof}

\subsection{Pre-corner theory}\label{s:precorner}

The failure of the standard procedure does not allow us to construct a BF$^2$V theory. It is, however, still possible to analyse the pre-corner structure. To complete the picture, along the pre-corner two form \eqref{e:precorner_two_form} we have to find the pre-corner action $\widetilde{S}^{\partial}$ and  an expression for a Hamiltonian vector field. Even if the two-form is degenerate, we can still get a pair $\widetilde{Q}^{\partial}$ and $\widetilde{S}^{\partial}$ satisfying $\iota_{\widetilde{Q}^{\partial}} \widetilde{\varpi}^{\partial}= \delta \widetilde{S}^{\partial}$, out of the boundary data. 

Before proceeding, let us recall the spaces on which the pre-corner fields are defined. In degree $-1$, we have $\gamma_m^\dag\in\Omega^{2}(\Gamma, \wedge^2\mathcal{V}))[-1]$ and $y_m^\dag\in\Omega^{2}(\Gamma, \wedge^4\mathcal{V}))[-1]$. In degree $1$, we have the ghosts parametrizing the gauge symmetries, $c \in\Omega^{0}(\Gamma, \wedge^2\mathcal{V}_\Gamma))[1]$, and the ones parametrizing the diffeomorphisms: respectively, $\xi \in\mathfrak{X}[1](\Gamma)$ tangential to $\Gamma$, $\xi^m\in \Omega^{0}(\Gamma)[1]$ transversal to $\Gamma$ into $\Sigma$, and  $\lambda\in \Omega^{0}(\Gamma)[1]$ transversal also to $\Sigma$.
In degree zero, we first have the tangent part $e \in \Omega_\text{nd}^1(\Gamma, \mathcal{V}_\Gamma)$  of the coframe restricted to the corner and its transversal part $e_m\in \Omega^0(\Gamma, \mathcal{V}_\Gamma)$, together with a fixed 
nowhere vanishing
field $\epsilon_n \in \Omega^0(\Gamma, \mathcal{V}_\Gamma)$ with the requirement that $eee_m\epsilon_n$ is nowhere $0$.\footnote{The fixed field $\epsilon_n$ and the still dynamical one $e_m$ may be interpreted as the two transversal components of the coframe, the latter being transversal with respect to the inclusion $\Gamma=\partial\Sigma\hookrightarrow\Sigma$ and the former with respect to the inclusion of $\Sigma$ as boundary of a bulk.} 
Furthermore, we also have a connection $\omega\in \mathcal{A}^\text{red}(\Gamma)$  where $\mathcal{A}^\text{red}(\Gamma)$ is the space of connections (on $P|_{\Gamma}$) such that the following equations are satisfied: 
\begin{align*}
\epsilon_n \dd_{\omega}e + &\gamma_m^{\dag}\widehat{Z}^m= e \sigma,\\
    e_m \sigma &\in \Ima W_{\partial\partial}^{(0,1)},
\end{align*}
where $Z=[c, \epsilon_n]+ L_{\xi}^{\omega} \epsilon_n+ \dd_{\omega_m} \epsilon_n \xi^m$.

\begin{remark}
These last equations are a consequence of the fact that the starting data on the boundary were constrained by \eqref{e:structural_constraint}; hence, this constraint will also descend to the pre-corner. However, it will split into two separate equations:
\begin{align*}
\epsilon_n \dd_{\omega}e + &\gamma_m^{\dag}\widehat{Z}^m= e \sigma, \\
\epsilon_n \dd_{\omega_m}e + \epsilon_n \dd_{\omega}e_m  + &\iota_{\widehat{Z}}\gamma_{m}^{\dag}= e_m \sigma + e \sigma_m.
\end{align*}
The second equation is dynamical but still gives some information about $\sigma$ and $\sigma_m$. In particular, we can rewrite it as 
\begin{align*}
    e_m \sigma \in \Ima W_{\partial\partial}^{(0,1)}.
\end{align*}
An interpretation of these constraints is given in Appendix \ref{sec:analysis_constraints}.
\end{remark}

\begin{remark}
    
    The map $\widehat{\bullet}$ has been defined in \eqref{e:hatmap} for fields on the boundary $\Sigma$. However, when we have combinations of the type $\iota_{\widehat{X}}\alpha$ for some form $\alpha$ on the boundary and some section $X$ of $\mathcal{V}_{\Sigma}$, we can pull them back to the corner and get $\iota_{\widehat{X}}\alpha+ \alpha_m\widehat{X}^m $.
\end{remark}

Let us now compute the pre-corner action.
Since we have the boundary cohomological vector field, we can let $\partial \Sigma= \Gamma \neq \emptyset$ and, using the modified master equation $\iota_{Q^{\partial}} \iota_{Q^{\partial}} \varpi^{\partial} = 2 \widetilde{S}^{\partial}$,   find an expression for the pre-corner action. After a long but straightforward computation, we get
\begin{align}\label{e:action-precorner}
\widetilde{S}^{\partial} &= \int_{\Gamma} \Big(\frac{1}{4}[c,c] ee + \frac{1}{2}\iota_\xi (ee) \dd_\omega c + e e_m \xi^m \dd_\omega c + \lambda \epsilon_n e \dd_\omega c \\
 &\phantom{=}+ \frac{1}{4}\iota_\xi \iota_\xi (ee) F_\omega + \iota_\xi e e_m \xi^m F_\omega + \iota_\xi e \epsilon_n \lambda F_\omega + e_m \xi^m \epsilon_n  \lambda F_\omega  \nonumber\\ 
 &\phantom{=}+  \frac{1}{2}[c,c] \gamma_m^{\dag} \xi^m + L_\xi^\omega c \gamma_m^{\dag} \xi^m + \frac{1}{2}\iota_\xi \iota_\xi  F_\omega \gamma_m^{\dag} \xi^m \nonumber\\ 
&\phantom{=}+ \frac{1}{2}\iota_{[\xi,\xi]}e y_m^\dag \xi^m + L_\xi^\omega(\lambda \epsilon_n) y_m^\dag \xi^m + L_\xi^\omega(e_m \xi^m) y_m^\dag \xi^m + [c, \lambda \epsilon_n]y_m^\dag \xi^m\Big). \nonumber
\end{align}

The last bit of information that is missing is a pre-corner cohomological vector field. This can be obtained by pushing forward the one on the boundary to the corner. We collect some technical lemmata that are useful for this computation in Appendix \ref{sec:push_forward_ham}.
\begin{remark}
Due to the degeneracy of the pre-corner two form, a Hamiltonian vector field defined through $\iota_{\widetilde{Q}^{\partial}} \widetilde{\varpi}^{\partial}= \delta\widetilde{S}^{\partial}$ is not unique  and  might differ from the projection of $Q^\partial$ by an element in the kernel of $\widetilde{\varpi}^\partial$.
\end{remark}

Collecting all the above information, we get the following expression for the pre-corner cohomological vector field $\widetilde{Q}^{\partial}$:
\begin{align*}
\widetilde{Q}^{\partial} e &= [c,e] + L_\xi^\omega e + \xi^m \dd_{\omega_m} e +e_m \dd \xi^m + \dd_{\omega}(\lambda \epsilon_n) + \lambda \sigma, \\
\widetilde{Q}^{\partial} e_m &= [c,e_m] + L_\xi^\omega e_m + \iota_{\partial_m \xi } e + \dd_{\omega_m}(e_m \xi^m) + \dd_{\omega_m}(\lambda \epsilon_n)+ \lambda \sigma_m, \\
\widetilde{Q}^{\partial} \omega &= \dd_\omega c - \iota_\xi F_\omega -F_{\omega_m} \xi ^m+ \lambda \mu+ \frac{1}{2}\lambda \epsilon_n \Lambda e,\\
\widetilde{Q}^{\partial} \omega_m &= \dd_{\omega_m} c - \iota_\xi F_{\omega_m} + \lambda \mu_m+ \frac{1}{2}\lambda \epsilon_n \Lambda e_m,\\
\widetilde{Q}^{\partial} c &= \frac{1}{2}[c,c] + \frac{1}{2}\iota_\xi\iota_\xi F_\omega+\iota_\xi F_{\omega_m}\xi^m +\lambda \iota_\xi \mu + \lambda \mu_m \xi^m,  \\
\widetilde{Q}^{\partial} \lambda &= Y^{(n)}, \\
\widetilde{Q}^{\partial} \xi &=\widehat{Y}+ \frac{1}{2}[\xi, \xi],\\
\widetilde{Q}^{\partial} \xi^m &=\widehat{Y}^{m}+ \frac{1}{2}[\xi, \xi]^m,\\
\widetilde{Q}^{\partial} \gamma^{\dag} &=  e_m \dd_\omega e +e \dd_{\omega_m} e +e \dd_\omega e_m +[c,\gamma_m^{\dag}] + L_\xi^\omega \gamma_m^{\dag} + \dd_{\omega_m} (\gamma_m^{\dag} \xi^m) + [\lambda \epsilon_n, y_m^\dag],\\
\widetilde{Q}^{\partial} y^\dag &= [c,y_m^\dag] + L_\xi^\omega y_m^\dag +\dd_{\omega_m} (y^\dag_m \xi^m)+ e_m F_\omega + e F_{\omega_m}+ \frac{1}{2} \Lambda e_m e^2 \nonumber \\ & {\qquad+ \lambda (\sigma_m y_m^\dag)^{(m)} + \lambda(\mu_m \gamma_m^{\dag})^{(m)}+ \lambda (\sigma_a y_m^\dag)^{(a)} + \lambda(\mu \gamma_{am}^{\dag})^{(a)}},
\end{align*}
where 
\begin{align*}
    Y&= [c, \lambda \epsilon_n] + L_\xi^\omega(\lambda \epsilon_n) + \xi^m \dd_{\omega_m}(\lambda \epsilon_n), \\ \mu &=(W_{\partial\partial}^{(1,2)})^{-1}(\epsilon_n F_{\omega}+y_m^\dag\widehat{Y}^{m}), \\
 \mu_m&= (W_{\partial\partial}^{(0,2)})^{-1}(e_m \mu + \epsilon_n F_{\omega_m}+\iota_{\widehat{Y}}y_{m}^\dag), & &
\end{align*}
and  $e_a Z_{a}^{(a)}= Z_{a}$, $e_m Z_{m}^{(m)}= Z_{m}$.
The data just collected do not form a BF$^2$V structure on the corner, since the closed two-form \eqref{e:precorner_two_form} is degenerate. Nonetheless, using the procedure described in Section \ref{s:generalization_Poisson}, it is possible to extract information from this structure.
\section{\texorpdfstring{$P_{\infty}$ structure of general pre-corner theory}{Poisson infinity structure of general pre-corner theory}}\label{s:Pstru}
As explained in Section \ref{sec:BF2Vstructure},  BF$^2$V theories define a $P_\infty$ structure once a polarization is chosen on  the space of corner fields. Furthermore (see Remark \ref{r:Poissondegenerateform}), this construction can be generalized to the cases when the two-form is degenerate, which is precisely the case at hand. In this section, we analyze these structures. In order to have a better understanding of the results that we find, we will afterward consider two simplified theories in Section \ref{s:simplified}, for which the structure will be more readable.

Since the two-form is not symplectic, we consider the construction explained in Remarks \ref{r:Hamiltoniandeg} and \ref{r:Poissondegenerateform}. Following the notation introduced in section \ref{sec:BF2Vstructure}, we consider a splitting of the Hamiltonian functionals and define $\mathfrak{h}$ to be a subalgebra of functionals in the variables $e, \xi, \lambda, \xi^m$ and $\gamma_m^\dag \xi^m$. {The projection to it is just obtained by setting $\omega=\omega_0$, a fixed background connection, and by putting to zero all the other fields}.\footnote{The reasons for the choice of these coordinates will be clearer later. Indeed, in one of the simplified cases, the tangent theory (Section \ref{sec:tangenttheory}), this choice corresponds to the generalization of a possible choice of polarization.}
In particular, we consider the following Hamiltonian functionals and prove that they form a $P_\infty$ subalgebra of $\mathfrak{h}$:

\begin{align*}
    J_{\varphi}&= \int_{\Gamma} \varphi \left(\frac{1}{2}ee+ \gamma_m^\dag \xi^m\right),\\
    M_{Y}&= \int_{\Gamma} Y \left(\iota_{\xi}\left(\frac{1}{2}ee+ \gamma_m^\dag \xi^m\right) + \alpha e\right),\\
    K_{Z}&= \int_{\Gamma} Z \left(\frac{1}{2}\iota_{\xi}\iota_{\xi}\left(\frac{1}{2}ee + \gamma_m^\dag \xi^m\right)+ \iota_{\xi}e \alpha+\frac{1}{2}\alpha^2\right),
\end{align*}
where $\alpha= \epsilon_n \lambda+e_m \xi^m$. These functionals are Hamiltonian because it is possible to construct the corresponding Hamiltonian vector fields, which read
\begin{align*}
    \mathbb{J}_{\varphi}&= \int_{\Gamma}\varphi \frac{\delta}{\delta c}, \\
    \mathbb{M}_Y&= \int_{\Gamma} Y \frac{\delta}{\delta \omega},\\
    \mathbb{K}_Z&= \int_{\Gamma} \Bigg(\left(-\iota_\xi Z + (W_{\partial\partial}^{(2,3)})^{-1}(\epsilon_n \lambda Z)\right)\frac{\delta}{\delta \omega}\\
    & \quad +\left(-\frac{1}{2}\iota_\xi \iota_\xi Z+ \iota_\xi(W_{\partial\partial}^{(2,3)})^{-1}(\epsilon_n \lambda Z)-(W_{\partial\partial}^{(2,3)})^{-1}(e_m \xi^m (W_{\partial\partial}^{(2,3)})^{-1}(\epsilon_n \lambda Z)) \right)\frac{\delta}{\delta c}\\
    & \quad +\left(e_m Z+ \gamma_m^\dag (W_{\partial\partial}^{(2,3)})^{-1}(e_m  (W_{\partial\partial}^{(2,3)})^{-1}(\epsilon_n \lambda Z))^{(m)}+ (W_{\partial\partial}^{(2,3)})^{-1}(\epsilon_n \lambda Z)\gamma_{am}^\dag)^{(a)} \right)\frac{\delta}{\delta y_m^\dag}\Bigg).
\end{align*}
We can then prove that they form a subalgebra by computing the various brackets. After a long but straightforward computation, we get the following result:
\begin{align*}
    &\{\}_0=\int_{\Gamma}\left(\frac{1}{2}\iota_{\xi}\iota_{\xi}\left(\frac{1}{2}ee + \gamma_m^\dag \xi^m\right)+ \iota_{\xi}e \alpha+\frac{1}{2}\alpha^2\right)F_{\omega_0}, \\
    &\{J_{\varphi}\}_1= M_{\dd_{\omega_0}\varphi}, &&
    \{M_{Y}\}_1 = K_{\dd_{\omega_0}Y}, &&
    \{K_{Z}\}_1 = 0, \\
    &\{J_{\varphi}, J_{\varphi'}\}_2 = J_{[\varphi, \varphi']}, &&
    \{J_{\varphi}, M_Y\}_2 = M_{[\varphi, Y]}, &&
    \{M_Y, K_Z\}_2 = 0,\\
    &\{M_Y, M_{Y'}\}_2 = K_{[Y, Y']}, &&
    \{J_{\varphi}, K_Z\}_2 = K_{[\varphi,Z]}, &&
    \{K_Z, K_{Z'}\}_2 = 0.
\end{align*}
 Note that the nullary operation is here obtained by the nonvanishing part of the projection of the action to $\mathfrak{h}$. We can write
\begin{align*}
    \{\}_0= K_{F_{\omega_0}},
\end{align*}
so the algebra generated by ${J}$, ${M}$, and ${K}$ closes also under the nullary operation. We also explicitly note that this structure is identical to the \emph{tangent theory} and that of $BF$ theory in \eqref{e:UVW_BF}.

\begin{remark}
As before, the similarity between the structure of the subalgebra of observables and that of $BF$ theory is connected to the possibility of obtaining the constrained theory as $BF$ theory for the Lie algebra $\mathfrak{so(3,1)}$,  restricted to the submanifold of fields parametrized by
 \begin{align*}
     c&=c, & A &= \omega, & B^\dag &= 0,\\
     \phi&= \frac{1}{4} \iota_\xi\iota_\xi(ee)+ \frac{1}{2}\iota_\xi\iota_\xi \gamma_m^\dag \xi^m 
     + \iota_\xi e \alpha + \frac{1}{2}\alpha^2, & \tau&= \frac{1}{2}\iota_\xi(ee)+ \iota_\xi\gamma_m^\dag \xi^m + e \alpha, & B &= \frac12 ee + \gamma_m^\dag \xi^m.
 \end{align*}
\end{remark}

\section{Simplified theories}\label{s:simplified}
{The expressions of the pre-corner data without reduction are rather complicated and the information contained in them is well hidden. For this reason, it is useful to consider some simplified cases in which the Poisson structure is more manifest. In this section, we propose two different simplified theory in which the physical content is more explicit. In the first, we impose some constraints on the boundary data, which do not change the on-shell boundary structure (i.e., we consider a smaller BFV theory still describing the same reduced phase space of the original one). In the second, we assume some ghost fields to vanish, thus not considering some symmetries (the diffeomorphisms normal to the boundary and to the corner).}

\subsection{Constrained theory} \label{sec:constrainedtheory}
This approach is based first on considering the BFV theory on a cylindrical boundary manifold (i.e., assuming $\Sigma=\Gamma\times I$, where $I$ is an interval, and then focusing on one of the two boundary components $\Gamma$). Next, we impose some further constraints, on the line of \eqref{e:structural_constraint}, to get a theory that is on-shell equivalent to the original one but better treatable with the BF$^2$V machinery. 

\begin{remark}
This approach is based on the fact that the failure of the two-form \eqref{e:precorner_two_form} to have a regular kernel has similar causes to the same failure of the pre-boundary two-form \cite{CS2017}. As discussed in \cite{CCS2020b}, it is anyway possible to overcome the problem by constructing a BV theory on the bulk with some additional constraints. Indeed, using the constraints suggested by the AKSZ construction, it is possible to construct a BV theory that induces a BFV theory on the boundary. 

We now want to mimic this behaviour in order to get a BFV theory that induces a BF$^2$V theory on the corner. Since we do not have at hand a corner theory, we cannot use any suggestion from the AKSZ construction and we can only try to guess the correct constraints.
\end{remark}

Assume that the manifold $\Sigma$ has the form of a cylinder, $\Sigma=\Gamma\times I$, and call $x^m$ the coordinate along $I$. Then,
a possible choice is given by the following constraints:
\begin{subequations}\label{e:constraints_corner}
\begin{align}
    &\gamma_m^\dag = e K, \label{e:constraint_corner1}\\
    &e_m \dd_{\omega} e + e_m \dd \xi^m K + \dd_{\omega}(\lambda \epsilon_n) K + \lambda\sigma K + [\lambda \epsilon_n, y_m^\dag]= eL,\label{e:constraint_corner2}\\
    &\epsilon_n K =0, \label{e:constraint_corner3b}\\
    &\epsilon_n L + \epsilon_n \dd_{\omega_m}e+ \epsilon_n \dd_{\omega} e_m + [c, \epsilon_n] K + L_{\xi}^{\omega}\epsilon_n K +\dd_{\omega_m} \epsilon_n  \xi^m K=0. \label{e:constraint_corner4b}
\end{align}
\end{subequations}

\begin{remark}
As we will see  later on, these constraints are sufficient to get a simplified version of the pre-corner structure, but they still do not grant the possibility of doing a  proper symplectic reduction. 
\end{remark}

\begin{remark}
 Note that these constraints do not modify the boundary theory, in the sense that the constraints do not modify the classical critical locus of the unconstrained theory described in Section~\ref{s:BFVPC}. Indeed, \eqref{e:constraint_corner1} and \eqref{e:constraint_corner3b} are constraints on an anti-field and have no meaning in the classical interpretation. On the other hand, \eqref{e:constraint_corner2} and \eqref{e:constraint_corner4b} encode part of the Euler--Lagrange equations on the boundary. To see this, we can rewrite the equation $e d_{\omega}e=0$ on the cylindrical boundary manifold $\Sigma=\Gamma\times I$ and get the equation
\begin{align*}
    e_m d_{\omega}e + e (d_{\omega})_m e + e d_{\omega}e_m=0.
\end{align*}
Since $W_{\partial\partial}^{1,1}$ is neither injective nor surjective, besides the dynamical equation describing $\partial_m e$ we get also 
\begin{align*}
    e_m d_{\omega}e = e L'
\end{align*}
for some $L'$. This last equation, modulo anti-fields (which can be ignored at the classical level), is the same as \eqref{e:constraint_corner2}. Then, \eqref{e:constraint_corner4b} is added to guarantee the invariance under the action of $Q^{\partial}$, as proved in Lemma~\ref{lem:constraints_and_Q} below.
\end{remark}

These constraints are fixing some components of the pre-corner fields $\omega$ and $\gamma_m^\dag$. Namely, we fix three components of $\omega$ in the kernel of $W_{\partial\partial}^{(1,2)}$ and four components of $\gamma_m^\dag$. More details can be found in \ref{sec:analysis_constraints} with the relevant proofs. 

\begin{remark}
These additional constraints on the boundary simplify the expression of the structural constraints \eqref{e:structural_constraint}. Dividing them into tangential and transversal to the corner, we obtain 
\begin{align*}
   \epsilon_n \dd_{\omega}e + \widehat{Y}^{m}e K &= e \sigma, \\
   \epsilon_n \dd_{\omega_m}e + \epsilon_n \dd_{\omega}e_m +\iota_{\widehat{Y}}(e K )&= e_m \sigma + e \sigma_m, 
\end{align*}
where $Y=[c, \epsilon_n]+ L_{\xi}^{\omega} \epsilon_n + \dd_{\omega_m}(\epsilon_n) \xi^m. $

\end{remark}

Furthermore, it is worth noting that since $W_{\partial\partial}^{(1,1)}$ is surjective, we can write $y_m^\dag= e x_m^\dag$ for some $x_m^\dag$. Moreover, since $W_{\partial\partial}^{(1,1)}$ is not injective, we can also ask that $\epsilon_n x_m^\dag = e A$ for some $A$. Indeed, this condition fixes only some components of $x_m^\dag$ in the kernel of $W_{\partial\partial}^{(1,1)}$.

\begin{lemma}\label{lem:constraints_and_Q}
The set of constraints \eqref{e:constraints_corner} is conserved under the action of $Q^{\partial}$, i.e., it is possible to define $Q^{\partial}K$ and $Q^{\partial}L$ so that 
\begin{align*}
   &Q^{\partial} \gamma_m^\dag = Q^{\partial}e K + e Q^{\partial}K, \\
   &\epsilon_n Q^{\partial}K =0, \\
   & Q^{\partial}(e_m \dd_{\omega} e + e_m \dd \xi^m K + \dd_{\omega}(\lambda \epsilon_n) K + \lambda\sigma K + [\lambda \epsilon_n, y_m^\dag])= Q^{\partial}eL + e Q^{\partial}L,\\
       &\epsilon_n Q^{\partial} L + Q^{\partial}(\epsilon_n \dd_{\omega_m}e+ \epsilon_n \dd_{\omega} e_m + [c, \epsilon_n] K + L_{\xi}^{\omega}\epsilon_n K +\dd_{\omega_m} \epsilon_n  \xi^m K)=0.
\end{align*}
\end{lemma}

\begin{proof}
We use the expressions of the components of $Q^{\partial}$ recalled in \eqref{e:Q-boundary}. We start from \eqref{e:constraint_corner1}. After a short computation, it is possible to see that $Q^{\partial} \gamma_m^\dag = Q^{\partial}e K + e Q^{\partial}K$ is satisfied modulo a term proportional to \eqref{e:constraint_corner2} by choosing 
\begin{align*}
    Q^{\partial}K = \dd_{\omega_m} e + \dd_{\omega}e_m + L_{\xi}^{\omega} K + [c,K] + \dd_{\omega_m} ( K \xi^m) +L + \mathbb{K},
\end{align*}
where $\mathbb{K} \in \Ker{W_{\partial\partial}^{(1,1)}}$ is not fixed by this equation. We use this freedom to choose a $Q^{\partial}K$ such that \eqref{e:constraint_corner2} is invariant as well. Indeed, it is a long but straightforward computation to show that \eqref{e:constraint_corner2} is invariant  and the correct choice for $Q^{\partial}K$ is with $\mathbb{K}=0$ and 
\begin{align*}
    Q^{\partial}L =& L_{\xi}^{\omega} L + [c,L] + \dd_{\omega_m} (L\xi^m) + \dd_{\omega}(\lambda\sigma_m)+[(\mathbb{V}_{\omega})_m,e]+[\mathbb{V}_{\omega},e_m] + \iota_{\partial_m \xi} \dd_{\omega} e +[\lambda\epsilon_n, (F_{\omega})_m]\\&+ \dd_{\omega_m}(\lambda \widehat{Y}^{m}K)+ \lambda \iota_{\widehat{Y}}(\dd_{\omega} K)+\iota_{\partial_m \xi} K \dd\xi^m + [((W_{\partial\partial}^{(1,2)})^{-1}(\lambda\epsilon_n F_{\omega}))_m,e] +\mathbb{L}\\& + [(W_{\partial\partial}^{(1,2)})^{-1}(\lambda\epsilon_n F_{\omega}),e_m]+ [(W_{\partial\partial}^{(1,2)})^{-1}(\lambda \widehat{Y}^{m}y_m^\dag),e_m]+[((W_{\partial\partial}^{(0,2)})^{-1}(\lambda \iota_{\widehat{Y}}y^\dag))_m,e]\\&+\dd_{\omega}(\lambda \iota_{\widehat{Y}}K), 
\end{align*}
where $\mathbb{L} \in \Ker{W_{\partial\partial}^{(1,1)}}$ is not fixed by this equation. Lastly, \eqref{e:constraint_corner3b} is invariant thanks to \eqref{e:constraint_corner4b}, which in turn is invariant by choosing $\epsilon_n\mathbb{L}=0$. 
\end{proof}

From the previous lemma we deduce that the constraints \eqref{e:constraints_corner} define a submanifold of $\mathcal{F}^{\partial}$ compatible with $Q^{\partial}$. As a consequence, they define a pre-BFV theory.

\subsubsection{Corner theory}
Starting from this new constrained BFV theory, it is possible to build a partial symplectic reduction on the new pre-corner two-form and to write the pre-corner symplectic form and the pre-corner action in more readable variables. First, we fix a section $\epsilon_m$  of $\mathcal{V}_{\Gamma}$ that is linearly independent from $\epsilon_n$, and we only allow fields $e$ that form a basis together with $\epsilon_m$ and $\epsilon_n$.  In other words, we have that the combination $ee\epsilon_m\epsilon_n \neq 0$ everywhere.
Next, we consider the map
\begin{align*}
    \te &= e + K \xi^m,\\
    \tom &= \omega + x_m^\dag \xi^m,\\
    \tc &= c + \iota_{\xi} x_m^\dag \xi^m +W^{-1}(\lambda\epsilon_n x_m^\dag \xi^m),\\
    \epsilon_{m} &= k^m e_m+ k^a e_a + k^n \epsilon_n,\\
    \txi{m} &= \frac{1}{k^m}\xi^m,\\
    \txi{a} &= \xi^a + \frac{k^a}{k^m}\xi^m,\\
    \widetilde{\lambda} &= \lambda + \frac{k^n}{k^m}\xi^m,
\end{align*}
where $k_a,k_n, k_m$ are functions, with $k_m\neq 0$ , chosen so that $\widetilde{Q}^\partial \epsilon_m=0$. 
The target space is then defined as the direct sum 
\begin{align*}
    \underbrace{\Omega_{\partial\partial\text{nd}}^{1,1}}_{\te} \oplus \underbrace{\mathcal{A}_{\text{red}}^{\partial\partial}}_{\tom}\oplus \underbrace{\Omega_{\partial\partial}^{0,2}[1]}_{\tc}\oplus \underbrace{\mathfrak{X}[1](\Gamma)}_{\txi{}}\oplus
    \underbrace{\Omega_{\partial\partial}^{0,0}[1]}_{\txi{m}}\oplus
    \underbrace{\Omega_{\partial\partial}^{0,0}[1]}_{\widetilde{\lambda}},
\end{align*}
where the fields must satisfy 
\begin{align*}
    \txi{m} \epsilon_m \dd_{\tom}\te + \widetilde{\lambda}\epsilon_n \dd_{\tom}\te = \te (\widetilde{\lambda}\widetilde{\sigma}+ \txi{m}\widetilde{L}),\\
    \txi{m} \epsilon_n \dd_{\tom}\te = \te \widetilde{\sigma}\txi{m},\\
    \txi{m} \epsilon_m \widetilde{\sigma} + \te \widetilde{\sigma}_m \txi{m} + \widetilde{L}\epsilon_n\txi{m}=0,
\end{align*}
for some $\widetilde{\sigma}\in \Omega_{\partial\partial}^{1,1}$, $\widetilde{\sigma}_m\in\Omega_{\partial\partial}^{0,1}$ and  $\widetilde{L}\in\Omega_{\partial\partial}^{1,1}$.

With these variables, the pre-corner two-form and the pre-corner action are, respectively,
\begin{align}\label{e:new_pre-corner_form}
    \widetilde{\varpi}^{\partial\partial}&= \int_{\Gamma} \left(\delta \tc \te \delta \te + \delta (\iota_{\txi{}}\te\te)\delta \tom + \delta (\epsilon_m\txi{m}\te)\delta \tom + \delta (\widetilde{\lambda} \epsilon_n \te)\delta \tom\right), \\
    \widetilde{S}^{\partial\partial} &= \int_{\Gamma}\Big(\frac{1}{4}[\tc,\tc]\te\te
    + \iota_{\txi{}}\te\te \dd_{\tom} \tc + \epsilon_m\txi{m}\te \dd_{\tom} \tc + \widetilde{\lambda} \epsilon_n \te \dd_{\tom} \tc \label{e:new_pre-corner_action}\\
    &\phantom{=} + \frac{1}{4}\iota_{\txi{}}\iota_{\txi{}}(\te\te) F_{\tom}+ \iota_{\txi{}}\te \epsilon_m\txi{m}F_{\tom}+  \iota_{\txi{}}\te\widetilde{\lambda} \epsilon_n F_{\tom}+ \epsilon_m\txi{m} \widetilde{\lambda} \epsilon_n F_{\tom}\Big). \nonumber
\end{align}

It is also possible to give an explicit expression of the cohomological vector field $\widetilde{Q}^{\partial\partial}$. This can be either be computed as the Hamiltonian vector field of the action $\widetilde{S}^{\partial\partial}$ or pushed forward from the boundary vector field ${Q}^{\partial}$. Both these methods lead to the following expression:
\begin{align*}
    \widetilde{Q}^{\partial\partial} \te &= [\tc , \te] + L_{\txi{}}^{\tom} \te + \dd_{\tom} (\epsilon_m \txi{m}+ \widetilde{\lambda} \epsilon_n) + \widetilde{\lambda} \widetilde{\sigma} + \widetilde{L}\txi{m}, \\
    \widetilde{Q}^{\partial\partial} \txi{m} &= X_m^{[m]} + X_n^{[m]} + \widetilde{\lambda}\widetilde{\sigma}_m^{[m]} \txi{m},\\
    \widetilde{Q}^{\partial\partial} \txi{a} &= X_m^{[a]} + X_n^{[a]} + \widetilde{\lambda}\widetilde{\sigma}_m^{[a]} \txi{m}+ \frac{1}{2}[\txi{}, \txi{}]^a,\\
    \widetilde{Q}^{\partial\partial}  \widetilde{\lambda}&= X_m^{[n]} + X_n^{[n]} + \widetilde{\lambda}\widetilde{\sigma}_m^{[n]} \txi{m},\\
    \widetilde{Q}^{\partial\partial} \tom &= \dd_{\tom} \tc - \iota_{\txi{}}F_{\tom} + (W_{\partial\partial}^{(1,2)})^{-1}((\epsilon_m \txi{m}F_{\tom}+\epsilon_n \widetilde{\lambda}F_{\tom}) + \mathbb{V}_{\tom},\\
    \widetilde{Q}^{\partial\partial} \tc &= \frac{1}{2} [\tc,\tc] + \frac{1}{2} \iota_{\txi{}}\iota_{\txi{}}F_{\tom} + \iota_{\txi{}}(W_{\partial\partial}^{(1,2)})^{-1}(\epsilon_m \txi{m}F_{\tom}+\epsilon_n \widetilde{\lambda}F_{\tom}) + \iota_{\txi{}}\mathbb{V}_{\tom}\\
     + &(W_{\partial\partial}^{(0,2)})^{-1}(\epsilon_m \txi{m}\mathbb{V}_{\tom}+\epsilon_n \widetilde{\lambda}\mathbb{V}_{\tom})+ (W_{\partial\partial}^{(0,2)})^{-1}((\epsilon_m \txi{m}+\epsilon_n \widetilde{\lambda})(W_{\partial\partial}^{(1,2)})^{-1}(\epsilon_m \txi{m}F_{\tom}+\epsilon_n \widetilde{\lambda}F_{\tom})),
\end{align*}
where $X_m= [\tc, \epsilon_m \txi{m}]+ L_{\txi{}}^{\tom}(\epsilon_m \txi{m})$, $X_n= [\tc, \epsilon_n \widetilde{\lambda}]+ L_{\txi{}}^{\tom}(\epsilon_n \widetilde{\lambda})$, $\widetilde{\sigma}= \sigma + X^{(m)}K+[\epsilon_n, x^{\dag}_m\xi^{m}]+[A\xi^m, \te]$, $\widetilde{L}= Lk^m + k^n \widetilde{\sigma}+k^a (d_{\tom}\te)_a$,  and $\widetilde{\sigma}_m=k^m \sigma_m +k^m X^{(a)}K_a + k^a \sigma_a + k^a X^m K_a$. The square brackets denote the components with respect to the basis $\{\te, \epsilon_m, \epsilon_n\}$, e.g., $X_m= X_m^{[a]}\te_a + X_m^{[m]}\epsilon_m+ X_m^{[n]}\epsilon_n$.\footnote{Note that using the properties of $e, \epsilon_m$ and $ \epsilon_n$, it would be possible to express these components without local coordinates, using the analogue of the map \eqref{e:hatmap} for the corner fields. Hence these expressions do not depend on the choice of local coordinates.}
Since the two form \eqref{e:new_pre-corner_form} is still degenerate (see below), the Hamiltonian vector field $\widetilde{Q}^{\partial\partial}$ is not unique, as it can be seen by the presence of inverses of maps $(W_{\partial\partial}^{(1,2)})$ which are not injective.

The two-form \eqref{e:new_pre-corner_form} is not symplectic. The equations defining its kernel are the following:
\begin{align*}
\delta \tc: &\quad  \te X_{\te}  =0, \\
\delta \te: &\quad  \te X_{\tc}  - \te \iota_{\txi{}} X_{\tom} - \widetilde{\lambda} \epsilon_n X_{\tom}- \epsilon_m \txi{m}  X_{\tom}=0,\\
\delta \txi{}: &\quad \te_{\bullet} \te X_{\tom}=0,\\
\delta \tom: &\quad - \iota_{X_{\txi{}}} \te \te - \iota_{\txi{}}( \te  X_{\te}) + \epsilon_m X_{\txi{m}} \te - \epsilon_m \txi{m} X_{\te} \nonumber\\ & \qquad -X_{\widetilde{\lambda}} \epsilon_n \te   -\widetilde{\lambda} \epsilon_n X_{\te}  =0,\\
\delta \txi{m}: &\quad \epsilon_m \te X_{\tom}=0,\\
\delta \widetilde{\lambda}: &\quad -\epsilon_n \te  X_{\tom}=0.
\end{align*}

We can simplify this system by noting that the third and the last two equations together form the equation $\te  X_{\tom}=0.$ 
Hence, it can be rewritten as
\begin{align*}
 &  \te X_{\te}  =0, \\
 &  \te (X_{\tc}  - \iota_{\txi{}} X_{\tom}) - (\widetilde{\lambda} \epsilon_n + \epsilon_m \txi{m})  X_{\tom}=0,\\
 &  \te X_{\tom}=0,\\
 &  \te ( -\iota_{X_{\txi{}}} \te   + \epsilon_m X_{\txi{m}} -X_{\widetilde{\lambda}} \epsilon_n) - (\epsilon_m \txi{m}   -\widetilde{\lambda} \epsilon_n ) X_{\te}  =0.
\end{align*}

This system is still singular since the map $W_{\partial\partial}^{(0,2)}$ appearing in the second equation is  neither injective nor surjective, and the map $W_{\partial\partial}^{(0,1)}$ appearing in the fourth is injective but not surjective.
However, it is worth noting that with the extra requests $(\widetilde{\lambda} \epsilon_n + \epsilon_m \txi{m})  X_{\tom}=0$ and $(\epsilon_m \txi{m}   -\widetilde{\lambda} \epsilon_n ) X_{\te}  =0$ we get $X_{\te}  =0$, $X_{\tom}=0$ from the first and the third equation, while the second identifies equivalence classes of $[c]$ and the fourth can be solved yielding $X_{\txi{}}$, $X_{\txi{m}}$ and $X_{\lambda}$.

\subsubsection{\texorpdfstring{$P_\infty$ structure }{Poisson infinity structure}}\label{s:constrained_Poisson_algebra}
Let us now analyze the $P_\infty$ structure of this constrained theory. Since the two-form is not symplectic, as in the general case we have to consider the construction explained in Section \ref{s:generalization_Poisson}. In order to keep the notation light, in this section we drop the tildes on the fields since no confusion can arise. The splitting that we consider here follows the one of the general theory described in Section \ref{s:Pstru}. Indeed, we define
$\mathfrak{h}$ to be a subalgebra of functionals in the variables $e, \xi, \lambda$ and $\xi^m$. As before, the projection to it is obtained by fixing $\omega$ to a background connection $\omega_0$  and by setting to zero all the other fields. The Hamiltonian functionals that we consider are again derived from the general case (we also use the same notation) and are the following:
\begin{align*}
    J_{\varphi}&= \int_{\Gamma} \frac12\varphi ee,\\
    M_{Y}&= \int_{\Gamma} Y (\iota_{\xi} e + \alpha)e,\\
    K_{Z}&= \int_{\Gamma} Z \left(\iota_{\xi}e\left(\frac{1}{2}\iota_{\xi} e + \alpha\right)+\frac{1}{2}\alpha^2\right),
\end{align*}
where $\alpha= \epsilon_n \lambda + \epsilon_m \xi^m$.\footnote{Note that here and in the following expression of the Hamiltonian vector fields, $\epsilon_m$ is fixed; hence, there is symmetry between the directions $m$ and $n$, while in the general case $e_m$ is a field of the theory and $\epsilon_n$ is fixed.}
These functionals are Hamiltonian because it is possible to construct the corresponding Hamiltonian vector fields, which read
\begin{align*}
    \mathbb{J}_{\varphi}&= \int_{\Gamma}\varphi \frac{\delta}{\delta c}, \\
    \mathbb{M}_Y&= \int_{\Gamma} Y \frac{\delta}{\delta \omega},\\
    \mathbb{K}_Z&= \int_{\Gamma}\Bigg( \left(-\iota_\xi Z + (W_{\partial\partial}^{(2,3)})^{-1}(\alpha Z)\right)\frac{\delta}{\delta \omega}\\
    & \quad +\left(-\frac{1}{2}\iota_\xi \iota_\xi Z+ \iota_\xi(W_{\partial\partial}^{(2,3)})^{-1}(\alpha Z)-(W_{\partial\partial}^{(2,3)})^{-1}(\alpha (W_{\partial\partial}^{(2,3)})^{-1}(\alpha Z)) \right)\frac{\delta}{\delta c}\Bigg).
\end{align*}
These functionals form a $P_\infty$ subalgebra of $\mathfrak{h}$ and the corresponding brackets read exactly as in the general case.

\begin{remark}
As in the general case, there is a similarity between the structure of the subalgebra of observables and that of $BF$ theory. 
\end{remark}

\subsection{Tangent theory}\label{sec:tangenttheory}

Let us now consider an even simpler case where we assume $\xi^m=0$ and $\lambda=0$ on the corner.\footnote{We call this theory \emph{tangent}\/ because we set to zero the transversal vector fields $\xi^m$ and $\lambda$ and we retain only the tangential vector field $\xi$.} As we will see, these two conditions are sufficient 
in order to get a regular kernel,  so we can perform a symplectic reduction and get a proper BF$^2$V theory. 

\begin{remark}
Note that assuming either only $\xi^m=0$ or only $\lambda=0$ is not sufficient to get a regular kernel. For example, considering the first case, we get that the pre-corner two-form becomes
\begin{align*}
    \widetilde{\varpi}_{part}^{\partial}= \int_{\Gamma} & (\delta c e \delta e - \iota_{\delta \xi} e e \delta \omega- \iota_{\xi}( e \delta e) \delta \omega -\delta \lambda \epsilon_n e  \delta \omega -\lambda \epsilon_n \delta e  \delta \omega )
\end{align*} 
on the space $\widetilde{\mathcal{F}}^{\partial}_{\text{part}}$ (given by the restriction to the corner of the fields appearing above).
The equations defining the kernel of the corresponding application $(\widetilde{\varpi}_{part}^{\partial})^{\sharp}$ are
\begin{subequations}
\begin{align}
    \delta c: &\quad  e X_e  =0, \label{e:kernel_Xe_part}\\
    \delta e: &\quad  e X_c -e \iota_\xi X_{\omega}-\lambda \epsilon_n X_{\omega}=0, \\
    \delta \xi: &\quad e_{\bullet} e X_{\omega}=0,\\
    \delta \omega: &\quad - \iota_{X_\xi} e e - \iota_{\xi}( e  X_e)-X_\lambda \epsilon_n e -\lambda \epsilon_n X_e  =0,\\
    \delta \lambda: &\quad -\epsilon_n e  X_\omega=0.
\end{align}
\end{subequations}
This system is still singular. Indeed, the third element of the second equation might not be proportional to $e$ and the map $W_{\partial\partial}^{(0,2)}$ is not surjective.
\end{remark}

Let us now consider, as announced, the case $\xi^m=0$ and $\lambda=0$; i.e., we retain only the tangential vector fields. 
The pre-corner two-form now reads
\begin{align*}
\widetilde{\varpi}_{\text{part}}^{\partial}= \int_{\Gamma} & \left(\delta c e \delta e - \iota_{\delta \xi} e e \delta \omega- \iota_{\xi}( e \delta e) \delta \omega\right).
\end{align*} 
The only remaining fields are those displayed in this formula. Note that, in particular, the transversal component $e_m$ of the coframe has disappeared.
The only remaining, open, condition is that $e \in \Omega^1(\Gamma, \mathcal{V}_\Gamma)$ should satisfy 
\begin{equation}\label{e:eeemen}
    ee\epsilon_m\epsilon_n\not=0,
\end{equation} where $\epsilon_m$ and  $\epsilon_n$ are fixed linearly independent sections of $\mathcal{V}_\Gamma$.\footnote{\label{f:VGamma}The dynamical field $e_m$ is now replaced by a fixed field $\epsilon_m$. Also note that, since $\mathcal{V}_\Gamma$ is assumed to arise as a restriction to $\Gamma$ from the boundary $\Sigma$, we are tacitly assuming that $\mathcal{V}_\Gamma$ is isomorphic to $T\Gamma\oplus\underline{\bbR}^2$.} 
The equations defining the kernel of the corresponding application $(\widetilde{\varpi}_{part}^{\partial})^{\sharp}$ are
\begin{subequations}
\begin{align*}
    \delta c: &\quad  e X_e  =0, \\
    \delta e: &\quad  e X_c -e \iota_\xi X_{\omega}=0, \\ 
    \delta \xi: &\quad e_{\bullet} e X_{\omega}=0,\\
    \delta \omega: &\quad - \iota_{X_\xi} e e - \iota_{\xi}( e  X_e)=0.
\end{align*}
\end{subequations}
This system is not singular.
Let us then define the following theory:
\begin{definition}\label{def:BFlikecornertheory}
We call \emph{$BF$-like corner theory} the BF$^2$V theory on the space of fields
\begin{align*}
    \check{\mathcal{F}}^{\partial\partial}= T^*[1]\left(\Omega_{\partial \partial}^{2,2}\oplus (\Omega_{\partial \partial}^{2,4} \otimes \Omega^1(\Gamma) )\right)
\end{align*}
with symplectic form
\begin{align*}
\check{\varpi}^{\partial\partial}= \int_{\Gamma} & \left(\delta \tc \delta \widetilde{E} - \iota_{\delta \widetilde{\xi}}  \delta \widetilde{P}\right)
\end{align*} 
and action
\begin{align*}
\check{S}^{\partial\partial} = \int_{\Gamma} &\left(\frac{1}{2}[\tc,\tc] \widetilde{E} + \iota_{\widetilde{\xi}} \widetilde{E} \dd_{\omega_0} \tc -  \frac{1}{2}\iota_{[\widetilde{\xi}, \widetilde{\xi}]} \widetilde{P} + \frac{1}{2}\widetilde{E} \iota_{\widetilde{\xi}}\iota_{\widetilde{\xi}}F_{\omega_0}\right),
\end{align*}
where $\omega_0$ is a reference connection.
\end{definition}

\begin{remark}
It is a straightforward check  that this is actually a BF$^2$V theory, i.e., that the action $\widetilde{S}^{\partial}$ satisfies the classical master equation.
\end{remark}

Furthermore, we can define a map  $\widetilde{\pi_\text{red}}:\widetilde{\mathcal{F}}^{\partial}\rightarrow\check{\mathcal{F}}^{\partial\partial}$:
\begin{equation*}
\widetilde{\pi_\text{red}} := \begin{cases}
\widetilde{E}= \frac{1}{2} e e &\\
\widetilde{c}= c + \iota_{\xi}(\omega-\omega_0) &\\
\widetilde{\xi}^i= \xi^i &\\
\widetilde{P}_i = \frac{1}{2} e e (\omega_i-\omega_{0i}) & 
\end{cases}
\end{equation*}
Notice that here we are assuming to work around a connection $\omega_0$.
It is a short computation to show that this map is compatible with the two-forms (respectively, the pre-corner form $\widetilde{\varpi}_{part}^{\partial}$ on $\widetilde{\mathcal{F}^{\partial}}$ and $\check{\varpi}^{\partial\partial}$ on $\check{\mathcal{F}}^{\partial\partial}$). 

Define now the submanifold $\mathcal{E} \subset \check{\mathcal{F}}^{\partial\partial}$ such that $(E,P,c,\xi)\in \mathcal{E}$ if $E$ is of the form $\frac{1}{2}ee$ for some $e$ satisfying $ee\epsilon_m\epsilon_n\not=0$, with $\epsilon_m$ and $\epsilon_n$ fixed linearly independent sections of $\mathcal{V}_\Gamma$ as above.\footnote{With a slight abuse of notation, we denote the fields in $\mathcal{E}$ with the same letter of those in $\widetilde{\mathcal{F}}^{\partial\partial}$ but without the tilde.} These conditions may be translated to requiring that the Pfaffian of $E$ vanishes 
and $E\epsilon_m\epsilon_n\neq 0$. In these cases, we drop the tilde.
As a consequence of the first statement of Proposition~\ref{p:finalprop}, which we prove in Appendix~\ref{s:appendixPfaffian}, 
$\mathcal{E}$ coincides with the image of $\widetilde{\pi_\text{red}}$. 

Let now $p': \Omega_{\partial\partial}^{0,2} \rightarrow \Omega_{\partial\partial}^{0,2}$ be a projection to the complement of the kernel of the map $W_{\partial\partial}^{(0,2)}: \Omega_{\partial\partial}^{0,2} \rightarrow \Omega_{\partial\partial}^{1,3}$. Then, the characteristic distribution of $\mathcal{E}$ is given by the vector fields $X_{p'c}$. Hence, we have the following

\begin{proposition}
The BF$^2$V space of  fields $\mathcal{F}^{\partial\partial}$ is symplectomorphic to the symplectic reduction of $\widetilde{\mathcal{F}}^{\partial}_{\text{part}}$.
\end{proposition}

We can express the symplectic form on the space of corner fields as  
\begin{align*}
\varpi^{\partial\partial}= \int_{\Gamma} & \left(\delta [c] \delta E - \iota_{\delta \xi}  \delta P\right),
\end{align*} 
where $E$ is a pure tensor as above and $[c]$ denotes the equivalence class of elements $c \in \Omega_{\partial \partial}^{0,2}[1]$ under the equivalence relation $c+ d \sim c$ for $d \in \Omega_{\partial \partial}^{0,2}[1]$ such that $ed=0$.

From the expression of the pre-corner action in this particular case, 
\begin{align*}
\widetilde{S}^{\partial} = \int_{\Gamma} &\left(\frac{1}{4}[c,c] ee + \frac{1}{2}\iota_\xi (ee) \dd_\omega c +  \frac{1}{4}\iota_\xi \iota_\xi (ee) F_\omega \right),
\end{align*}
we can deduce the corresponding action on the corner:
\begin{align*}
S^{\partial\partial}_{\omega_0} = \int_{\Gamma} &\left(\frac{1}{2}[[c],[c]] E + \iota_{{\xi}} (E) \dd_{\omega_0} [c] -  \frac{1}{2}\iota_{[{\xi}, {\xi}]} {P} + \frac{1}{2}E \iota_{{\xi}}\iota_{{\xi}}F_{\omega_0}\right).
\end{align*}
This expression is invariant under the quotient map above: $\frac{1}{2}[c,c] ee = [ce,c]e - [e,c] e c= [ce,ce]$, $ \iota_\xi (ee) d c= -d\iota_\xi e e  c=  L_{\xi}(ee) c = 2 (L_{\xi}e) ec$.

\begin{remark}\label{r:openE}
The open condition $E\epsilon_m\epsilon_n\neq 0$ may possibly be dropped to get an extended version of the tangent corner theory (this is analogous to the observation that in $2+1$ PC gravity one may extend the theory dropping the condition that the coframe be nondegenerate). One might want, however, to retain the weaker open condition $E\not=0$ to ensure that the closed condition $\operatorname{Pf}(E)=0$ still defines a submanifold.
\end{remark}

\begin{remark}
The map $\pi_\text{red}$ is not strictly speaking the reduction with respect to the kernel of the pre-corner two-form but does satisfy the BV-BFV axioms. 
\end{remark}

\subsubsection{\texorpdfstring{$P_\infty$ structure }{Poisson infinity structure}}\label{s:Poissontangent}

We start our analysis of the $P_\infty$ structure of the tangent theory. Since it is a proper BF$^2$V theory, we can apply the results Section  \ref{sec:BF2Vstructure}.

We first study the structure of the $BF$-like corner theory as in Definition \ref{def:BFlikecornertheory}, and then, we give an implicit description of the corner Poisson structure of gravity by means of a quotient with respect to a suitable ideal.  Note that in this section we will drop the tilde on the fields, since no confusion can arise.

The case at hand is similar to that of $BF$ theory. The first step is to choose a polarization and reinterpret the space of fields as a cotangent bundle. We will consider two interesting polarizations.
\subsubsection{The first polarization}
Here, we choose the space of fields as the cotangent bundle of the space $\mathcal{N}$ with coordinates $E$ and $\xi$ and choose $\{P=c=0\}$\footnote{Choosing $P=0$ is equivalent to choose $\omega=\omega_0$ where $\omega_0$ is a reference connection.} as the Lagrangian submanifold. From the action, we get $\pi= \pi_0+ \pi_1 + \pi_2$ with
\begin{align*}
    \pi_0&= \int_\Gamma \frac{1}{2} E \iota_{\xi} \iota_{\xi} F_{\omega_0},\\
    \pi_1&= \int_\Gamma \left(\iota_{\xi} E \dd_{\omega_0} \frac{\delta}{\delta E} - \frac{1}{2}\iota_{[\xi,\xi]}\frac{\delta}{\delta \xi}\right),\\
    \pi_2&= \int_\Gamma \frac{1}{2}\left[\frac{\delta}{\delta E},\frac{\delta}{\delta E}\right] E. 
\end{align*}
These  equip $C^{\infty}(\mathcal{N})$ with the structure of a curved $P_\infty$ algebra. Note that this polarization roughly corresponds to the choice of subalgebra $\mathfrak{h}$ that we have made for the general and constrained theory in Sections \ref{s:Pstru} and \ref{s:constrained_Poisson_algebra}. Indeed, we consider  a subalgebra of linear functionals of the form\footnote{Once again, we use here the same notation for the functionals as in the general and constrained cases.}:
\begin{align*}
    J_{\varphi}&= \int_{\Gamma} \varphi E,\\
    M_{Y}&= \int_{\Gamma} Y \iota_{\xi} E,\\
    K_{Z}&= \int_{\Gamma} \frac{1}{2}Z \iota_{\xi}\iota_{\xi} E.
\end{align*}
The derived brackets are as follows
\begin{align*}
    &\{\}_0= \int_\Gamma \frac{1}{2} E \iota_{\xi} \iota_{\xi} F_{\omega_0}, && && \\
    &\{J_{\varphi}\}_1= M_{\dd_{\omega_0}\varphi}, &&
    \{M_{Y}\}_1 = K_{\dd_{\omega_0}Y}, &&
    \{K_{Z}\}_1 = 0, \\
    &\{J_{\varphi}, J_{\varphi'}\}_2 = J_{[\varphi, \varphi']}, &&
    \{J_{\varphi}, M_Y\}_2 = M_{[\varphi, Y]}, &&
    \{J_{\varphi}, K_Z\}_2 = K_{[\varphi,Z]},\\
    &\{M_Y, M_{Y'}\}_2 = K_{[Y, Y']}, &&
    \{M_Y, K_Z\}_2 = 0, &&
    \{K_Z, K_{Z'}\}_2 = 0.
\end{align*}
Observe the similarity with \eqref{e:UVW_BF} in $BF$ theory. Also, note that we can write
\[
\{\}_0= K_{F_{\omega_0}},
\]
so the algebra generated by $J$, $M$, and $K$ closes also under the nullary operation.
\begin{remark}
The striking similarity between the structure of the subalgebra of observable proposed in the present section and that of $BF$ theory is not accidental. In fact, the tangent theory (before the reduction) can be obtained as $BF$ theory, for the Lie algebra $\mathfrak{so(3,1)}$,  restricted to the submanifold of fields parametrized by
 \begin{align*}
     c&=c, & A &= \omega, & B^\dag &= 0,\\
     \phi&= \frac{1}{4} \iota_\xi\iota_\xi(ee), & \tau&=\frac{1}{2} \iota_\xi(ee), & B &=\frac{1}{2} ee.
 \end{align*}
\end{remark}

We now want to describe the $P_\infty$ structure of the real theory describing gravity. Hence, we have to consider the structure described above and assume that the Pfaffian of $E$ vanishes. Instead of describing it directly, we can describe the subalgebra  as the quotient of this $P_\infty$ algebra by the ideal generated by the following additional linear functionals:
\begin{align*}
    P_{\mu}&= \int_{\Gamma} \mu \mathcal{P}_E,\\
    Q_{\nu}&= \int_{\Gamma} \nu \iota_{\xi} \mathcal{P}_E,\\
    R_{\sigma}&= \int_{\Gamma} \frac12 \sigma \iota_{\xi}\iota_{\xi} \mathcal{P}_E,
\end{align*}
where $\mathcal{P}_E= \sqrt{\operatorname{Pf}(E)}$ is the square root of the Pfaffian of $E$.\footnote{Given the definition of Pfaffian in Appendix \ref{s:appendixPfaffian}, here $\mathcal{P}_E$ is a density defined as $$\mathcal{P}_E= \sqrt{\frac18\epsilon_{abcd}E_{12}^{ab}E_{12}^{cd}}\;\dd x^1\dd x^2.$$} It is worth noting that $\mathcal{P}_E$ is invariant under the action of the gauge transformations. Now we have to compute the brackets of these new linear functionals to show that they form an ideal of the $P_\infty$ algebra generated by $J, M, K, P, Q$ and $R$. Let us start from the 1-brackets. They read
\begin{align*}
    &\{P_{\mu}\}_1= Q_{\dd_{\omega_0}\mu}, &&
    \{Q_{\nu}\}_1 = R_{\dd_{\omega_0}\nu}, &&
    \{R_{\sigma}\}_1 = 0.
\end{align*}
On the other hand, all the 2-brackets containing $P,Q$ or $R$ vanish.

Hence, we can describe the $P_\infty$ algebra of such linear functionals on the space of corner fields in the \emph{tangent theory} as the quotient of the $P_\infty$ algebra generated by $J, M, K, P, Q$ and $R$ by the $P_\infty$ ideal generated by $P, Q$ and $R$.
\subsubsection{The second polarization}\label{s:secondpol}
We can now consider another polarization: we choose the space of fields as the cotangent bundle of the space $\mathcal{N}$ with coordinates $E$ and $P$ and choose $\{\xi=c=0\}$ as the Lagrangian submanifold. From the action, we get $\pi= \pi_2$ with
\begin{align*}
    \pi_2&= \int_\Gamma\left( \frac{1}{2}\left[\frac{\delta}{\delta E},\frac{\delta}{\delta E}\right] E  + \iota_{\frac{\delta}{\delta P}} (E) \dd_{\omega_0} \frac{\delta}{\delta E} -  \frac{1}{2}\iota_{[\frac{\delta}{\delta P}, \frac{\delta}{\delta P}]} {P} + \frac{1}{2}E \iota_{\frac{\delta}{\delta P}}\iota_{\frac{\delta}{\delta P}}F_{\omega_0}\right),
\end{align*}
which equips $C^{\infty}(\mathcal{N})$ with the structure of a Poisson algebra. As before, we can consider a subalgebra of linear functionals. Let \begin{align*}
  F_{X}= \int_{\Gamma} \iota_X P \quad \text{and} \quad  J_{\varphi}&= \int_{\Gamma} \varphi E.
\end{align*}
Their binary operations  are as follows:
\begin{align}\label{e:brackets_tangent_EP}
    \{J_{\varphi}, J_{\varphi'}\}_2 = J_{[\varphi, \varphi']}, &&
    \{J_{\varphi}, F_X\}_2 = J_{\iota_X \dd_{\omega_0} \varphi}, &&
    \{F_X, F_{X'}\}_2 = F_{[X, X']}+ J_{\iota_X\iota_{X'}F_{\omega_0}}.
\end{align}
As before, in order to get the structure on the gravity theory, we have to consider the ideal generated by the functional $P_{\mu}= \int_{\Gamma} \mu \mathcal{P}_E$. The only nonzero bracket is the one with $F_X$:
\begin{align*}
    \{P_{\mu}, F_X\}_2= P_{\iota_X \dd_{\omega_0}\mu}.
\end{align*}

It is worth noting that, with this polarization, the structure of linear functionals corresponds to that of (a subalgebra of) an Atiyah algebroid. 

\begin{remark}
    An Atiyah algebroid is a way of describing a combination of a Lie algebra of internal symmetries with the Lie algebra of vector fields (the infinitesimal version of a combination of a Lie group of symmetries with the group of diffeomorphisms). From this point of view, it is not unexpected that an Atiyah algebroid should appear as part of the corner structure (and in fact it also appears in some work completed after ours \cite{CL2023}). What is interesting is the explicit way in which it emerges from our construction.
\end{remark}

The goal of next section is to show this relation.

\subsubsection{Atiyah algebroids}\label{s:atiyah_tangent}
Let us begin with some definitions.

\begin{definition}
Let $M$ be a manifold. A \textbf{Lie algebroid} over $M$ is a triple $(A, [\cdot, \cdot], \rho)$ where $A \rightarrow M$ is a vector bundle over $M$, $[\cdot, \cdot]\colon \Gamma(A) \times \Gamma(A) \rightarrow \Gamma(A)$ an $\mathbb{R}$-Lie bracket, and $\rho\colon A \rightarrow TM$ a morphism of vector bundles, called the anchor, such that
\begin{align*}
    [X, g Y]= \rho (X) g \cdot Y + g[X, Y] \qquad \forall X,Y \in \Gamma(A), \; g \in C^{\infty}(M).
\end{align*}
\end{definition}

The Atiyah algebroid is a particular example of a Lie Algebroid.

\begin{definition}
Let $G$ be a Lie Group and $P \rightarrow M$ a $G$-principal bundle over $M$. The \textbf{Atiyah algebroid} is a Lie Algebroid with 
$A = TP/G$, the Lie bracket on  sections that inherited from the tangent Lie algebroid of $P$, and the anchor  induced by the quotient by $G$ of the differential map $\dd\pi\colon TP \rightarrow TM$.
\end{definition}

The Atiyah algebroid may be written in terms of the short exact sequence
\begin{align*}
    0 \rightarrow \operatorname{ad} P \rightarrow A \rightarrow TM \rightarrow 0.
\end{align*}
The algebroid that we will construct out of the corner data will be of type $A= F \oplus TM$, corresponding to a splitting of the exact sequence. By well known results, this corresponds to a map $\tau\colon TM \rightarrow A$ such that $\pi \circ \tau = \text{id}_{TM}$. Out of this map, we can construct an isomorphism between $A$ and $F \oplus TM$ as follows:
\begin{align*}
    \chi:  F \oplus TM & \rightarrow A\\
     ( a, X ) & \mapsto  \iota(a)+ \tau (X).
\end{align*}
This map is injective. Indeed, let $\chi(a,X)=0$, then $\pi(\chi(a,X))=X=0$. As a consequence,  $\iota(a)=0 $ implying $a=0$.

Using this isomorphism, we can induce an algebroid structure on $F \oplus TM$. After a short computation, we find the following structure:
\begin{align*}
    [(a,X),(b,Y)]=([a,b]+ \iota^{-1}([\iota(a),\tau(Y)]+[\tau(X),\iota(b)]+[\tau(X),\tau(Y)]-\tau[X,Y]),[X,Y])
\end{align*}

We can now introduce the map $\nabla^\tau$
\begin{align*}
    \nabla^\tau \colon \Gamma(TM) \times \Gamma (F) &\rightarrow \Gamma(F)\\
    (X,a) & \mapsto \nabla^\tau_X(a) = \iota^{-1}([\iota(a),\tau(X)])
\end{align*}

\begin{lemma}
The map $\nabla^\tau$ has the following properties:
\begin{enumerate}
    \item $\nabla^\tau$ is a connection for $F$.
    \item The curvature of $\nabla^\tau$ is given by
    \begin{align*}
        R^\tau (X,Y)= \iota^{-1}([\tau(X),\tau(Y)]-\tau[X,Y]).
    \end{align*}
\end{enumerate}
\end{lemma}

\begin{proof}
Easy computation.
\end{proof}

Let us now denote by $\omega_0$ the connection one-form corresponding to the connection $\nabla^\tau$. Then, we can rewrite the brackets on $F \oplus TM$ as

\begin{align}\label{e:brackets_split_Atiyah}
    [(a,X),(b,Y)]=([a,b] - \iota_X \dd_{\omega_0} (b) + \iota_Y \dd_{\omega_0} (a) + \iota_X \iota_Y F_{\omega_0} ,[X,Y]).
\end{align}

The Lie algebroid structure on $A$ allows us to define a Poisson bracket on $\Gamma(A^*)$. We write this down for linear functionals. Namely, we define $U_\beta=\int_M\Phi\beta$, with $\Phi\in\Gamma(A^*)$ and $\beta\in\Gamma(A)$.
We then define
\begin{align*}
    \left\lbrace \int_{M}\Phi \beta_1,\int_{M}\Phi \beta_2\right\rbrace
    = \int_{M}\Phi [\beta_1,\beta_2].
\end{align*}
Let us now write $\Phi=\mathcal{F}+\mathcal{Q}$ with $\mathcal{F}\in \Gamma (\bigwedge^{\text{top}} T^*M, F^*)$ and  $\mathcal{Q}\in \Gamma (\bigwedge^{\text{top}} T^*M, T^*M)$. Then, using \eqref{e:brackets_split_Atiyah} we get
\begin{align}\label{e:brackets_algebroid}
    \left\lbrace \int_{M}(\mathcal{F}a + \mathcal{Q}X) ,\int_{M}(\mathcal{F}b + \mathcal{Q}Y)\right\rbrace
    = \int_{M}\!\!(\mathcal{F}([a,b] - \iota_X \dd_{\omega_0} (b) + \iota_Y \dd_{\omega_0} (a) + \iota_X \iota_Y F_{\omega_0})+ \mathcal{Q}[X,Y]).
\end{align}

\begin{theorem}
The BF$^2$V structure of the tangent theory on a corner $\Gamma$ induces an Atiyah algebroid structure on $\operatorname{ad} P \oplus T \Gamma$. 
\end{theorem}

\begin{proof}
Let us define $B=\operatorname{ad}P \oplus T \Gamma$. Then, the space of corner fields is $\mathcal{F}^{\partial\partial}= T^*[1]\Gamma (B)^*$.
As explained in the previous section, we can equip this space with a Poisson structure. Comparing \eqref{e:brackets_algebroid} with \eqref{e:brackets_tangent_EP}, it is easy to see that on linear functionals these brackets coincide  with the identification $E= \mathcal{E}$ and $P=\mathcal{Q}$. Hence, dualizing, the induced structure is the one of an Atiyah algebroid.
\end{proof}

\begin{remark}
This construction does not depend on the final quotient. Hence, the symplectic space of corner fields identifies a Poisson subalgebra and consequently a subalgebroid.
\end{remark}

\subsubsection{Quantization} \label{s:quantization_tangent}
In the relatively simple tangent case, we may also describe the quantization of the corner structure for a very important particular situation that arises when we  consider a point defect on a spacelike boundary $\Sigma$. By quantization, we mean here the deformation quantization of a Poisson manifold \cite{BFFLS} (solved in the finite-dimensional case by Kontsevich \cite{K2003}).\footnote{On general grounds, one expects to associate with a corner $\Gamma$ an algebra $A_\Gamma$ that quantizes the underlying Poisson structure. If $\Gamma$ occurs as the boundary of $\Sigma$, one expects $\Sigma$ to be quantized as an $A_\Gamma$-module $H_\Sigma$. Locality is then the statement that cutting some $\Sigma$ along a hypersurface $C$ into manifolds $\Sigma_1$ and $\Sigma_2$ with boundary $C$ should yield $A_\Gamma$-modules $H_{\Sigma_1}$ and $H_{\Sigma_2}$ associated with them such that $H_\Sigma=H_{\Sigma_1}\otimes_{A_\Gamma}H_{\Sigma_2}$.} We
take $\Gamma$ to be an infinitesimal sphere surrounding this point. On $\Gamma$, we only consider uniform  
fields (this is our formalization of its being infinitesimal). For $\xi$, which is a 
 vector field, this implies $\xi=0$. Similarly, we get $P=0$.
 In the resulting
theory, there are then no $\xi$ nor $P$. On the other hand, $c$ and $E$ are $\operatorname{SO}(3)$-equivariant. Since the BF$^2$V action and $2$-form are defined in terms of an invariant pairing, what matters are only the values of $c$ and $E$ at some point. We denote the first as $c\in\Lambda^2V$ and the second as
$E=\mathbf{A}\,\text{vol}$, with $\mathbf{A}\in\Lambda^2V$ and 
$\text{vol}$ the standard, normalized volume form on the sphere $\Gamma$ evaluated at the chosen point. 
We then have the symplectic form
\begin{align*}
    \varpi^{\partial\partial}_q= \delta c\, \delta \mathbf{A}
\end{align*}
and the BF$^2$V action
\begin{align*}
    S^{\partial\partial}_q= \frac12 [c,c]\,\mathbf{A}.
\end{align*}
(Note that both expressions take values in $\Lambda^4V$ which we tacitly identify with $\mathbb{R}$.)
Next, we will have to impose that $E$ is a pure tensor satisfying $E\epsilon_m\epsilon_n\not=0$ for some fixed linearly independent sections $\epsilon_m$ and $\epsilon_n$ in $V$. This corresponds to imposing $\operatorname{Pf}(\mathbf{A})=0$ and $\mathbf{A}\epsilon_m\epsilon_n\not=0$, and to reduce $c$ accordingly. Note that the second condition on $\mathbf{A}$ is an open condition, which, in particular, entails $\mathbf{A}\not=0$.

We first analyze the theory without the conditions on $\mathbf{A}$. In the polarization $c=0$, the above data yield as  Poisson manifold 
the dual of the Lie algebra $\frg=\mathfrak{so}(3,1)\simeq\Lambda^2V$. Its quantization may be identified with the universal enveloping algebra $U(\frg)$ of $\frg$ \cite{G1983}.
The module structures for $\Sigma$ minus the defect that we get from the quantization of the corner then correspond to representations of $U(\frg)$, but this is the same as Lie algebra representations of $\frg$ or group representations of its simply connected Lie group $G=\mathrm{SL}(2,\mathbb{C})$.

The conditions on $\mathbf{A}$  
select a five-dimensional Poisson submanifold 
of $\frg^*$.
Since $\operatorname{Pf}(\mathbf{A})$ is quadratic in $\mathbf{A}$ and invariant,
it is a quadratic Casimir.
If we ignore the open condition $\mathbf{A}\epsilon_m\epsilon_n\not=0$,
the quantization then simply amounts to considering representations of $G$ in which this Casimir is represented as zero.
Explicitly
we write
\[
\mathbf{A}=\begin{pmatrix}
0             & A^{01}  & A^{02}   & A^{03}\\
-A^{01}   &   0         & A^{12}   & A^{13}\\
-A^{02}   & -A^{12} &   0         & A^{23}\\
-A^{03}   & -A^{13} &  -A^{23} &   0
\end{pmatrix}=:
\begin{pmatrix}
0             & M^{1}  & M^{2}   & M^{3}\\
-M^{1}   &   0         & J^{3}   & -J^{2}\\
-M^{2}   & -J^{3} &   0         & J^{1}\\
-M^{3}   & J^{2} &  -J^{1} &   0
\end{pmatrix}.
\]
We then have
\[
\operatorname{Pf}(\mathbf{A})
= A^{01}A^{23}-A^{02}A^{13}+A^{03}A^{12}
= \mathbf{M}\cdot\mathbf{J}
= \frac{\mathbf{J}_+^2-\mathbf{J}_-^2}4,
\]
with $\mathbf{J}_\pm=\mathbf{J}\pm\mathbf{M}$. Note that $\mathbf{J}_\pm^2$ are the two standard $\mathfrak{su}(2)$ quadratic Casimirs of the two summands of $\mathfrak{sl}(2,\mathbb{C})=\mathfrak{su}(2)\oplus\mathfrak{su}(2)$. 
The condition $\operatorname{Pf}(\mathbf{A})=0$, i.e., $\mathbf{J}_+^2=\mathbf{J}_-^2$,
therefore implies that we only have
representations of $\mathrm{SO}(3,1)^+$ with highest weight of the form $(m,m)$ (here, $2m$ is a nonnegative integer). 

The open condition $\mathbf{A}\epsilon_m\epsilon_n\not=0$ is more difficult to understand algebraically. The induced open condition
$\mathbf{A}\not=0$ instead corresponds to $\mathbf{J}_+^2\not=0$ and $\mathbf{J}_-^2\not=0$, which would suggest that we have to exclude the case $m=0$.
On the other hand, it might make sense to retain also this possibility in the quantization (essentially working with the extended theory of Remark~\ref{r:openE}).

To summarize the results of this section, we see that, in the case of small $m$, the point defect then corresponds to a scalar ($m=0$), a vector ($m=\frac12$), and a traceless symmetric tensor ($m=2$).

\section{Cosmological term}\label{s:cosmological}
In the previous sections, we have always assumed the vanishing of the cosmological constant. 
We now drop this assumption and add the following term to the boundary BFV action:
\begin{align*}
    S^{\partial}_{\text{cosm}}= \int_{\Sigma}\frac{1}{6}\Lambda \lambda \epsilon_n e^3.
\end{align*}
Since it does not contain any derivatives, this additional term does not change the pre-corner two-form \eqref{e:precorner_two_form} and hence, the extendability of the BFV theory to a BFV-BF$^2$V theory. The only change in the pre-corner structure is an additional term in the pre-corner action \eqref{e:action-precorner} of the form
\begin{align*}
    \widetilde{S}^{\partial}_{\text{cosm}}= \int_{\Gamma}\frac{1}{2}\Lambda \lambda \epsilon_n \xi^m e_m e^2 .
\end{align*}

Since this term contains $\xi^m$, the tangent case is unmodified and carries no information about the cosmological constant. 

However, the action of the constrained case \eqref{e:new_pre-corner_action} gets a contribution of the form
\begin{align*}
    \widetilde{S}^{\partial}_{\text{cosm}}= \int_{\Gamma}\frac{1}{2}\Lambda \widetilde\lambda \epsilon_n \txi{m} \epsilon_m \te^2 .
\end{align*}
In the constrained case and in the pre-corner case, there are some differences when the cosmological constant is present, similarly to what happens in $BF$ theory. Indeed, even though the unary operation $\{\ \}_1$ and the binary operation $\{\ ,\ \}_2$ do not change,
we have  
\begin{align*}
    \{\}_0&=\int_{\Gamma}\left(\iota_{\txi{}}\te\left(\frac{1}{2}\iota_{\txi{}} \te + \alpha\right)+\frac{1}{2}\alpha^2\right)F_{\omega_0} + \int_{\Gamma}\frac{1}{2}\Lambda \lambda \epsilon_n \xi^m \epsilon_m \te^2,\\
    \{\}_0&=\int_{\Gamma}\left(\iota_{\xi}e\left(\frac{1}{2}\iota_{\xi} e + \alpha\right)+\frac{1}{2}\alpha^2\right)F_{\omega_0} + \int_{\Gamma}\frac{1}{2}\Lambda \lambda \epsilon_n \xi^m e_m e^2,
\end{align*}
for the constrained and the pre-corner theories, 
where $\alpha$ is as defined in Sections \ref{s:constrained_Poisson_algebra} and \ref{s:Pstru}, respectively. As a result, the algebra generated by ${J}$, ${M}$, and ${K}$ no longer closes under the nullary operation. To remedy for this, 
we can add a functional $C_{\beta}$ to the $P_\infty$ subalgebra to parametrize this new term as follows\footnote{We spell the details in the pre-corner case. In the constrained case, it is just sufficient to add a tilde to the variables and to change the expression of $\alpha$ to get the required functionals. The brackets hold verbatim.}:
\begin{align*}
   C_{\beta}&= \int_{\Gamma} \frac{1}{2}\beta  ee\alpha^2.
\end{align*}
We now have
\[
\{\}_0= K_{F_{\omega_0}}+C_\Lambda.
\]
In order to get a closed set under the bracket operations, we also  add the following two additional functionals:
\begin{align*}
   D_{\gamma}&= \int_{\Gamma} \frac{1}{2}\gamma  \iota_{\xi}(ee)\alpha^2,\\
   E_{\rho}&= \int_{\Gamma} \frac{1}{4}\rho  \iota_{\xi}\iota_{\xi}(ee)\alpha^2.
\end{align*}
The brackets of these functionals with themselves and with $J_{\varphi}$, $M_y$, $K_Z$ are all zero except for
\begin{align*}
    \{C_\beta\}_1=D_{\dd\beta} \qquad \{D_\gamma\}_1=E_{\dd\gamma}.
\end{align*}

\appendix
\section{Notation and property of maps}\label{s:appendix_notation}
The goal of this appendix is to recall and collect in one place the relevant quantities and maps, to establish the conventions, and to summarize the technical results needed in the article.

Let us first recall some useful shorthand notation introduced in the previous sections. Let $M$ be a smooth manifold of dimension $4$ with corners, and let us denote by $\Sigma= \partial M$ its $3$-dimensional boundary and by $\Gamma= \partial\partial M$ its $2$-dimensional corner. Furthermore, we will use the  notation $\mathcal{V}_{\Sigma}$  for the restriction of $\mathcal{V}$ to $\Sigma$ and $\mathcal{V}_{\Gamma}$  for the restriction of $\mathcal{V}$ to $\Gamma$. We  define 
\begin{align*}
     \Omega_{\partial}^{i,j}:= \Omega^i\left(\Sigma, \textstyle{\bigwedge^j} \mathcal{V}_{\Sigma}\right),
    \qquad \Omega_{\partial\partial}^{i,j}:= \Omega^i\left(\Gamma, \textstyle{\bigwedge^j} \mathcal{V}_{\Gamma}\right).
\end{align*}

On $\Omega_{\partial}^{i,j}$ and $\Omega_{\partial\partial}^{i,j}$, we define the following maps:
\begin{align*}
    W_{\partial}^{(i,j)}\colon \Omega_{\partial}^{i,j}  & \longrightarrow  \Omega_{\partial}^{i,j}\\ 
    X  & \longmapsto  X \wedge e|_{\Sigma} ,
    \\ 
    W_{\partial\partial}^{  (i,j)}\colon \Omega_{\partial\partial}^{i,j}  & \longrightarrow  \Omega_{\partial\partial}^{i,j}\\ 
    X  & \longmapsto  X \wedge e|_{\Gamma} .
\end{align*}
\begin{remark}
Usually, we will omit writing the restriction of $e$ to the corresponding manifold $\Sigma$ or $\Gamma$.
\end{remark}

The properties of these maps are collected in the following lemmata, where we  condensate all the information in two tables, one for the boundary maps and one for the corner maps.  We organize the $\Omega_{\bullet}^{ij}$ spaces in an array and connect them with arrows corresponding to the maps $W_{\bullet}^{\bullet(i,j)}$: a hooked arrow denotes an injective map, while a two-headed arrow denotes a surjective map. 
 The proofs of these properties are similar to those proved in \cite{CCS2020} and are left to the reader. 
 
On the boundary, the index $i$ runs only between 1 and 3. 

\begin{lemma}\label{lem:Wmaps_properties_boundary}
The maps $W_{\partial}^{(i,j)}$ on the boundary fields have the properties described in the following table:
\begin{center}
\begin{equation}
\begin{tikzcd}[ row sep= 2 em, column sep= 3 em]
\Omega_{\partial}^{0,0} \arrow[rd, hook]& \Omega_{\partial}^{0,1} \arrow[rd, hook]& \Omega_{\partial}^{0,2} \arrow[rd, hook]& \Omega_{\partial}^{0,3} \arrow[rd, two heads]& \Omega_{\partial}^{0,4} \\
\Omega_{\partial}^{1,0} \arrow[rd, hook]& \Omega_{\partial}^{1,1} \arrow[rd, hook]& \Omega_{\partial}^{1,2} \arrow[rd, two heads]& \Omega_{\partial}^{1,3} \arrow[rd, two heads]& \Omega_{\partial}^{1,4} \\
\Omega_{\partial}^{2,0} \arrow[rd, hook]& \Omega_{\partial}^{2,1} \arrow[rd, two heads]& \Omega_{\partial}^{2,2} \arrow[rd, two heads]& \Omega_{\partial}^{2,3} \arrow[rd, two heads]& \Omega_{\partial}^{2,4} \\
\Omega_{\partial}^{3,0} & \Omega_{\partial}^{3,1} & \Omega_{\partial}^{3,2} & \Omega_{\partial}^{3,3} & \Omega_{\partial}^{3,4} 
\end{tikzcd}
\end{equation}
\end{center}
\end{lemma}
\begin{lemma}\label{lem:Wmaps_properties_corner}
The maps $W_{\partial\partial}^{(i,j)}$ on the corner fields have the properties described in the following table:
\begin{center}
\begin{equation}\label{e:cod2We}
\begin{tikzcd}[ row sep= 2 em, column sep= 3 em]
\Omega_{\partial\partial}^{0,0} \arrow[rd, hook]& \Omega_{\partial\partial}^{0,1} \arrow[rd, hook]& \Omega_{\partial\partial}^{0,2} \arrow[rd]& \Omega_{\partial\partial}^{0,3} \arrow[rd, two heads]& \Omega_{\partial\partial}^{0,4} \\
\Omega_{\partial\partial}^{1,0} \arrow[rd, hook]& \Omega_{\partial\partial}^{1,1} \arrow[rd]& \Omega_{\partial\partial}^{1,2} \arrow[rd, two heads]& \Omega_{\partial\partial}^{1,3} \arrow[rd, two heads]& \Omega_{\partial\partial}^{1,4} \\
\Omega_{\partial\partial}^{2,0} & \Omega_{\partial\partial}^{2,1} & \Omega_{\partial\partial}^{2,2} & \Omega_{\partial\partial}^{2,3} & \Omega_{\partial\partial}^{2,4}
\end{tikzcd}
\end{equation}
\end{center}
\end{lemma}

The coframe $e$ viewed as an isomorphism $e\colon TM \rightarrow \mathcal{V}$ defines, given a set of coordinates on $M$, a preferred basis on $\mathcal{V}$. If we denote by $\partial_i$ the vector field in $TU$, where $U$ is a coordinate neighborhood in $M$, corresponding to the coordinate $x^i$, we get a basis on $\mathcal{V}|_{U}$ by $e_i:= e (\partial_i)$. On the boundary, since $T\Sigma$ has one dimension less than $\mathcal{V}_{\Sigma}$, we can complement the linear independent set $(e_i)$ with another independent vector that we will call $\epsilon_n$. On the corner $\Gamma$, the tangent space loses one further dimension; hence, we will have to introduce one more additional independent vector that will be denoted by $\epsilon_m$. 
Fixed a coordinate system on $M$ (or $\Sigma$ or $\Gamma$), we  call this basis the \emph{standard basis} and, unless otherwise stated, the components of the fields will always be taken with respect to this basis.

\section{Pfaffian and pure tensors}\label{s:appendixPfaffian}
In this appendix, we discuss the relation between having $\operatorname{Pf}(E)=0$ for an element $E \in \Omega_{\partial\partial}^{2,2}$ and requiring that $E$ can be expressed as a pure tensor, i.e., that $E=\frac{1}{2}ee$ for some $e \in \Omega_{\partial\partial}^{1,1}$. 
We start with the local analysis. Let
\[
\phi\colon\begin{array}[t]{ccc}
V\times V &\to &\Lambda^2V\\
(e_1,e_2)&\mapsto &e_1e_2
\end{array}
\]
where $V$ is a four-dimensional vector space and, as usual, we omitted the wedge multiplication symbol on the right hand side.
We then have the following two lemmata.
\begin{lemma}
$
e_1,\, e_2\ \text{linearly independent} \iff \phi(e_1,e_2)\not=0.
$
\end{lemma}
\begin{proof}
 If $e_1$ and $e_2$ are linearly independent, then we can complete them to a basis $\{e_1,e_2,e_3,e_4\}$, and we clearly have that $\phi(e_1,e_2)e_3e_4=e_1e_2e_3e_4\not=0$ as an element of $\Lambda^4V$, so $\phi(e_1,e_2)\not=0$. If, on the other hand, $e_1$ and $e_2$ are linearly dependent, then we have $e_1=\alpha e_2$ or $e_2=\alpha e_1$, for some scalar $\alpha$, so $e_1e_2=0$.
\end{proof}
\begin{lemma}
$\operatorname{Pf}(\phi(e_1,e_2))=0$ for all $e_1,e_2$.
\end{lemma}
\begin{proof}
For $E=(E^{ab})$ in some basis, we have
\[
\operatorname{Pf}(E)=\frac18\epsilon_{abcd}E^{ab}E^{cd}.
\]
Therefore, if $E^{ab}=e_1^ae_2^b-e_2^ae_1^b$, we clearly have $\operatorname{Pf}(E)=\frac12\epsilon_{abcd}e_1^ae_2^be_1^ce_2^d=0$.
\end{proof}

A further interesting remark is that, for $E=e_1e_2$, we have $Ee_1=Ee_2=0$. This can also be written in terms of matrix multiplication if we introduce $\Check E:=*E\in\Lambda^2V^*$, i.e., $\Check E_{ab}=\epsilon_{abcd} E^{cd}$. Now, we have $\Check E\cdot e_1=\Check E\cdot e_2 = 0$. For further reference, we also introduce the linear map
$\psi_E\colon V\to V^*$, $v\mapsto\Check E\cdot v$.

Let us finally introduce
\[
W:=\{(e_1,e_2)\in V\times V \ |\ e_1,\, e_2\ \text{linearly independent}
\}
\]
and
\[
B:=\{E\in \Lambda^2 V\setminus\{0\}\ |\ \operatorname{Pf}(E)=0\}.
\]
For every $E\in B$, we define $\Check E=*E\in\Lambda^2V^*$ as above and the corresponding linear map
$\psi_E\colon V\to V^*$. 
\begin{lemma}
The kernel of $\psi_E$ is two-dimensional.
\end{lemma}
\begin{proof}
Since the matrix representing $E$ or $*E$ is skew-symmetric, its eigenvalues are either equal to zero or they come in pairs of conjugate nonzero imaginary numbers. Since $E\not=0$, they cannot all vanish. On the other hand, the condition $\operatorname{Pf}(E)=0$ implies that $E$ and $*E$ are singular; therefore, at least one eigenvalue must vanish. It then follows that exactly two eigenvalues vanish, whereas the other two are conjugate nonzero imaginary numbers.
\end{proof}
Let $S_E:=\operatorname{ker}\psi_E$.
\begin{lemma}
Let $(e_1,e_2)$ be a basis of $S_E$. Then, there is a uniquely determined nonzero scalar $\lambda$ such that
$E=\lambda e_1e_2$.
\end{lemma}
\begin{proof}
Let $E':=e_1e_2$. Then $S_{E'}=S_E$. Let us complete $(e_1,e_2)$ to a basis $(e_1,e_2,e_3,e_4)$ of $V$.
In this basis, we then have $\Check E_{1a}=\Check E'_{1a}=0$ and $\Check E_{2a}=\Check E'_{2a}=0$ for every $a$. By skew-symmetry, we then have that the only nonzero entries of $\Check E$ and $\Check E'$ are the 34 and the 43 ones, one opposite to the other. There is then a uniquely determined nonzero scalar $\lambda$ such that $E_{34}=\lambda E'_{34}$.
\end{proof}
Collecting all the above, we then have the
\begin{proposition}
$\phi(W)=B$.
\end{proposition}
\begin{proof}
For every $E\in B$, we can choose a basis $(e_1,e_2)$ of $S_E$ and we then have $E=\lambda e_1e_2$.
But then $(\lambda e_1,e_2)\in W$ and $E=\phi(\lambda e_1,e_2)$.
\end{proof}

The map $\phi$ is clearly not injective. We can, however, relate this to a distribution that is the same as the one that we get from the kernel of the two-form in the tangent corner structure, see \eqref{e:kernel_Xe_part}. Namely, let $K\subset TW$ be the regular involutive distribution spanned by vector fields $X=(X_1,X_2)$ satisfying $e_1X_2+X_1e_2=0$ (wedge product symbols omitted).
It is clear that $\phi$ is constant along $K$. Let $\underline\phi$ be the induced map $W/K\to B$.
\begin{proposition}
$\underline\phi$ is a diffeomorphism.
\end{proposition}
\begin{proof}
We have already seen that every $E\in B$ is of the form $E=\phi(e_1,e_2)$ with $(e_1,e_2)$ of $S_E$ a basis of $S_E$. Choose an inner product on $S_E$ and a reference vector $v\not=0$. By moving along $K$ (with $X_1=0$ and $X_2=e_1$), we can always arrange $e_1$ and $e_2$ to be orthogonal. By further moving along $K$ (with $X_1=e_1$ and $X_2=-e_2$), we can arrange $e_1$ and $e_2$ to have the same length. 

Now, suppose that $E=\phi(e_1,e_2)=\phi(e'_1,e'_2)$. By the above discussion, we may assume that $e_1$, $e_2$, $e'_1$, and $e'_2$ have the same length, that
$e_1$ is orthogonal to $e_2$, that $e'_1$ is orthogonal to $e'_2$, and that the two pairs have the same orientation on $S_E$. We can now rotate the vectors $e_1$ and $e_2$ (by choosing $X_1=e_2$ and $X_2=-e_1$) to send $e_1$ to $e'_1$. This automatically sends $e_2$ to $e'_2$.
\end{proof}

To get in touch with the corner structure, we need one more piece of information to implement condition \eqref{e:eeemen}; namely, the datum of two linearly independent vectors $\epsilon_m$ and $\epsilon_n$ in $V$. We then define
\[
W':=\{(e_1,e_2)\in V\times V \ |\ (e_1,\, e_2,\,\epsilon_m,\,\epsilon_n)\ \text{linearly independent}
\}\subset W
\]
and
\[
B':=\{E\in \Lambda^2V\ |\ E\epsilon_m\epsilon_n\not=0\text{ and }
\operatorname{Pf}(E)=0\}\subset B.
\]
Note that $W'$ is an open subset of $W$ and $B'$ is an open subset of $B$. It is immediately clear that $\phi(W')\subseteq B'$. On the other hand, if $E\in B'\subset B$, we can write $E=e_1e_2$. The condition $E\epsilon_m\epsilon_n\not=0$ implies that $e_1,\, e_2,\,\epsilon_m,\,\epsilon_n$ are linearly independent, so $(e_1,e_2)\in W'$. Moreover, the $K$-leaf of $(e_1,e_2)\in W'$ is contained in $W'$, as it has image a fixed $E\in B'$. Therefore, we have the following
\begin{proposition}
$\phi(W')=B'$, and $\underline\phi\colon W'/K\to B'$ is a diffeomorphism.
\end{proposition}

We finally move to the setting of the corner structure. The data are the following: a two-manifold $\Gamma$, a rank-four vector bundle $\mathcal{V}_\Gamma$ over $\Gamma$, which is assumed to be isomorphic to $T\Gamma\oplus\underline\bbR^2$, and two linearly independent sections $\epsilon_m,\epsilon_n$ of the $\underline\bbR^2$ summand of $\mathcal{V}_\Gamma$. 
We consider the map 
\begin{align*}
    \phi\colon   \Omega_{\partial\partial}^{1,1}:=\Gamma(T^*\Gamma\otimes \mathcal{V}_\Gamma)  & \to   \Gamma(\Lambda^2T^*\Gamma\otimes \Lambda^2\mathcal{V}_\Gamma)=:\Omega_{\partial\partial}^{2,2}\\
      e &  \mapsto     \frac12ee
\end{align*}
In local coordinates, we write $e=e_1\dd x^1+e_2\dd x^2$, so $E=\phi(e)=-e_1e_2\dd x^1\dd x^2$, which is the same map $\phi$ (up to the density $-\dd x^1\dd x^2$) that we considered in the first part of this section when we restrict ourselves to a fiber of $\mathcal{V}_\Gamma$.

We then define
\[
\mathcal{W}':=\{e\in\Omega_{\partial\partial}^{1,1} \ |\ ee\epsilon_m\epsilon_n\not=0
\}
\]
and
\[
\mathcal{B}':=\{E\in \Omega_{\partial\partial}^{2,2}\ |\ E\epsilon_m\epsilon_n\not=0\text{ and }
\operatorname{Pf}(E)=0\}.
\]
\begin{proposition}\label{p:finalprop}
$\phi(\mathcal{W}')=\mathcal{B}'$, and $\underline\phi\colon \mathcal{W}'/\mathcal{K}\to \mathcal{B}'$ is an isomorphism of fiber bundles where $\mathcal{K}$ is a distribution fiberwise defined as $K$.
\end{proposition}
\begin{proof}
Fiberwise we follow the proofs of the first part of this appendix. The only problem is to prove that globally we can write $E\in \mathcal{B}'$ as $\frac12ee$. The point is that the condition $E\epsilon_m\epsilon_n\not=0$ implies that the distribution of two-planes $S_E$ is transversal to the distribution $S_{\epsilon_m\epsilon_n}$, i.e., the $\underline\bbR^2$ summand of $V$. This means that for a given isomorphism $e^{0}$ of $T\Gamma$ with a complement of the $\underline\bbR^2$ summand (chosen in such a way that $e^0e^0\epsilon_m\epsilon_n$ defines the same orientation as $E\epsilon_m\epsilon_n$), 
we have $E=\frac12ee$ with $e$ of the form $fe^{0}+\alpha \epsilon_m +\beta \epsilon_n$, with $\alpha,\beta$ 1-forms on $\Gamma$ and $f$ a nowhere vanishing function.
\end{proof}

\section{Analysis of the constraints} \label{sec:analysis_constraints}
In this appendix, we analyze the constraints \eqref{e:constraints_corner} and show which fields are they fixing. Let us start with some preliminary results.
Consider $W_{\partial \partial}^{(1,2)}\colon \Omega_{\partial \partial}^{1,2}  \longrightarrow \Omega_{ \partial\partial}^{2,3}$.
The dimensions of domain and codomain are 
$\dim \Omega_{ \partial\partial}^{1,2} = 12$ and $\dim \Omega_{ \partial\partial}^{2,3} = 4$. The kernel of $W_{\partial \partial}^{(1,2)}$ is defined by 
\begin{align*}
X_{\mu_1}^{ab} e_a e_b 
e_{\mu_2} 
\cdots 
e_{\mu_{2}} \dd x^{\mu_1}\dd x^{\mu_2}\cdots \dd x^{\mu_{2}}=0,
\end{align*}
where we used $e_a$ as a basis for $\mathcal{V}_\Gamma$.\footnote{For simplicity of notation, we assume $\epsilon_n=e_4 $. The proof does not depend on this assumption.}  Since $\dd x^{1}\dd x^{2}$ is a basis for $\Omega^{2}(\Gamma)$, we obtain one equation of the form
\begin{align*}
 X_{1}^{ab} e_a e_b  e_{2} -  X_{2}^{ab} e_a e_b  e_{1} =0.  
\end{align*}
Recall now that $e_a e_b e_{\mu}$ for $\mu=1,2$ is a basis of $\wedge^{3}\mathcal{V}_\Gamma$. Hence, we obtain the following equations:
\begin{align*}
& X_1^{13}+X_2^{23}=0, & \qquad X_1^{14}+X_2^{24}=0,\\
& X_1^{34}=0, &   X_2^{34}=0.
\end{align*}
Hence, the map $W_{\partial \partial}^{(1,2)}$ is surjective but not injective. In particular, $\dim \text{Ker}W_{\partial \partial}^{(1,2)}=8$ and the kernel is generated by the following components:

\begin{align*}
& X_1^{13}-X_2^{23}, &  X_1^{14}-X_2^{24}, & & X_1^{12}, & &  X_2^{12},\\
& X_1^{23}, &   X_2^{13}, &
& X_1^{24}, & &  X_2^{14}.
\end{align*}

Consider now $\psi_e\colon \Omega_{\partial \partial}^{1,2} \rightarrow \Omega_{\partial \partial}^{2,1}$, $\psi_e(v):=[v,e]$.
The components of $\psi_e$ are defined by
\footnote{Here, we use that at every point we can find a basis in $\mathcal{V}_\Gamma$ such that $e_\mu^i= \delta_\mu^i$: $[v,e]_{\mu_1\mu_2}^a= v_{\mu_1}^{ab}\eta_{bc}e_{\mu_2}^c- v_{\mu_2}^{ab}\eta_{bc}e_{\mu_1}^c= v_{\mu_1}^{ab}e_{b}^d\eta_{dc}e_{\mu_2}^c- v_{\mu_2}^{ab}e_{b}^d\eta_{dc}e_{\mu_1}^c$.}
\begin{align*}
[v,e]_{\mu_1\mu_2}^a= v_{\mu_1}^{ab} g^{\partial\partial}_{b \mu_2} - v_{\mu_2}^{ab}g^{\partial\partial}_{b\mu_1}=0.
\end{align*}
Using now normal geodesic coordinates, we can diagonalize $g^{\partial\partial}$ with eigenvalues on the diagonal $\alpha_{\mu} \in \{1,-1,0\}$:
\begin{align*}
[v,e]_{\mu_1\mu_2}^a= v_{\mu_1}^{a\mu_2} \alpha_{ \mu_2} - v_{\mu_2}^{a\mu_1} \alpha_{\mu_1}.
\end{align*}

Let us now assume that $g^{\partial\partial}$ is nondegenerate and in particular space-like ($\alpha_{\mu}=1$). Then, the components of $\psi_e$ are defined by
\begin{align*}
[v,e]_{12}^1= v_{1}^{12},  & & [v,e]_{12}^3=v_{1}^{32}  - v_{2}^{31}, \\
[v,e]_{12}^2= v_{2}^{12},  & & [v,e]_{12}^4=v_{1}^{42}  - v_{2}^{41}. 
\end{align*}

We can now analyze part of the constraints \eqref{e:constraints_corner}. At the beginning, we just consider the classical part of them (i.e., we assume $c=\xi=\xi^m=\lambda=0$). The results will then straightforwardly generalize to the complete case. 

\begin{lemma}
The constraints
\begin{align*}
    \epsilon_n \dd_{\omega}e =e\sigma, && \epsilon_n \dd_{\omega_m}e+\epsilon_n \dd_{\omega}e_m=e\sigma_m + e_m \sigma, \\
    e_m \dd_{\omega}e=eL, && \epsilon_n L +e_m \sigma + e \sigma_m=0,
\end{align*}
fix four components of $\omega$.
\end{lemma}
\begin{proof}
Let us start with the restriction of the boundary constraint to the corner:
$\epsilon_n \dd_{\omega}e=\epsilon_n \dd e + \epsilon_n [\omega, e] =e\sigma$.
Let us denote $Y=\dd e$. Then using the results of the previous lemmata, we get that this equation translates into the following equations for components of the fields:
\begin{align*}
\omega_1^{32}-\omega_2^{31}= Y_{12}^{3},  & & \sigma_{2}^{4}=\omega_1^{12} + Y_{12}^1,&& \sigma_{1}^{4}=-\omega_2^{12} + Y_{12}^2,\\
\sigma_{1}^{3}=0,  & & \sigma_{2}^{3}=0, && \sigma_{1}^{1}+\sigma_{2}^{2}=0.
\end{align*}

The part transversal to the corner of the boundary structural constraint is
$\epsilon_n \dd_{\omega_m}e+\epsilon_n \dd_{\omega}e_m=e\sigma_m + e_m \sigma$. On the corner, it is a dynamical equation but also introduces some relations between the components of $\sigma$ and $\sigma_m$. These are
\begin{align*}
\sigma_{m}^{2}=0,  & & \sigma_{m}^{1}=0, && \sigma_{1}^{2}=0,\\
\sigma_{2}^{1}=0,  & & \sigma_{m}^{3}+\sigma_{1}^{1}=0, && \sigma_{m}^{3}+\sigma_{2}^{2}=0.
\end{align*}

In a similar way, we get the following equations for the components from the equation $e_m \dd_{\omega}e=e_m \dd e + e_m [\omega, e] =eL$:
\begin{align*}
\omega_1^{24}-\omega_2^{14}= Y_{12}^{4},  & & L_{2}^{3}=\omega_1^{12} + Y_{12}^1,&& L_{1}^{3}=-\omega_2^{12} + Y_{12}^2,\\
L_{1}^{4}=0,  & & L_{2}^{4}=0, && L_{1}^{1}+L_{2}^{2}=0.
\end{align*}

Lastly, we consider the constraint $\epsilon_n L +e_m \sigma + e \sigma_m=0.$ In components, we obtain some equations proportional to the previous ones and the following:
\begin{align*}
\sigma_{1}^{4}+L_1^3=0,  & & \sigma_{2}^{4}+L_2^3=0, && L_{1}^{2}=0,\\
L_{2}^{1}=0,  & & \sigma_{m}^{4}-L_{1}^{1}=0, && \sigma_{m}^{4}-L_{2}^{2}=0.
\end{align*}

Collecting all the information, we get the following equations for the components of $\omega:$
\begin{align*}
    \omega_1^{32}-\omega_2^{31}= Y_{12}^{3} && \omega_1^{24}-\omega_2^{14}= Y_{12}^{4} && \omega_1^{12} + Y_{12}^1=0 && \omega_2^{12} + Y_{12}^2=0.
\end{align*}
\end{proof}

To generalize this result to the case where also the ghosts are present, it is sufficient to modify the definitions of $\sigma, \sigma_m, L$, and $Y$. The components fixed will not change, but they will be fixed to a different combination of the other fields.

Let us now consider the two constraints $\gamma_m^{\dag}= eK$ and $\epsilon_n K=0.$ 
\begin{lemma}
The constraints \eqref{e:constraint_corner1} and \eqref{e:constraint_corner3b} fix four components of the field $\gamma_m^{\dag}$.
\end{lemma}
\begin{proof}
In components, \eqref{e:constraint_corner1} corresponds to the following relations:
\begin{align*}
    (\gamma_m^{\dag})_{12}^{12}=K_1^1+ K_2^2, && (\gamma_m^{\dag})_{12}^{13}=K_2^3, && (\gamma_m^{\dag})_{12}^{14}= K_2^4,\\
    (\gamma_m^{\dag})_{12}^{23}= -K_1^3, && (\gamma_m^{\dag})_{12}^{24}=-K_1^4, && (\gamma_m^{\dag})_{12}^{34}=0.
\end{align*}
On the other hand, \eqref{e:constraint_corner3b} correspond to the following relations:
\begin{align*}
    K_1^1=0, && K_1^3=0, && K_1^2=0, && K_2^1=0, && K_2^3=0, && K_2^2=0.
\end{align*}
Hence, combining the two sets of equations, we get four equations for the components of $\gamma_m^{\dag}$: 
\begin{align*}
    (\gamma_m^{\dag})_{12}^{12}=0, && (\gamma_m^{\dag})_{12}^{13}=0, && (\gamma_m^{\dag})_{12}^{23}=0, && (\gamma_m^{\dag})_{12}^{34}=0.
\end{align*}
\end{proof}

\section{Results about the push-forward of Hamiltonian vector fields}\label{sec:push_forward_ham}
In this appendix, we present some technical results that are useful to push-forward the Hamiltonian vector field $Q^{\partial}$ from the boundary to the corner.
Since the expression \eqref{e:Q-boundary} of $Q^{\partial}$  contains nonexplicit terms involving the function $(W_{\partial}^{(i,j)})^{-1}$, we must find a way to invert it. 
\begin{lemma}
Let $\widetilde{\gamma}\in \Omega_{\partial}^{i,j}$ and $\widetilde{X} \in \Omega_{\partial}^{i+1,j+1}$ be such that $\widetilde{\gamma}= (W_{\partial}^{(i,j)})^{-1}(\widetilde{X})$. If we let $\widetilde{e}=e|_{\Gamma} + e_m \dd x^m$, $ \widetilde{\gamma}= \gamma|_{\Gamma} + \gamma_m \dd x^m$, and $\widetilde{X}= X|_{\Gamma}+ X_m \dd x^m$, then  we have
\begin{align*}
    \gamma|_{\Gamma} &= (W_{\partial\partial}^{(i,j)})^{-1}(\pi_I (X|_{\Gamma})),\\
    \gamma_m &= (W_{\partial\partial}^{,(i-1,j)})^{-1}(\pi_I(-e_m (W_{\partial\partial}^{(i,j)})^{-1}(\pi_I (X|_{\Gamma})) + X_m)).
\end{align*}
\end{lemma}
\begin{proof}
Omitting the restriction to the corner, we have that
\begin{align*}
 \widetilde{e} \widetilde{\gamma}= (e + e_m \dd x^m)(\gamma + \gamma_m \dd x^m) = X+ X_m \dd x^m = \widetilde{X}.
\end{align*}
This equation splits into two subequations, containing $\dd x^m$ or not:
\begin{align*}
 e \gamma = X, \qquad \qquad e \gamma_m + e_m \gamma = X_m.
\end{align*}
From the first, we deduce $\gamma = (W_{\partial\partial}^{(i,j)})^{-1}(\pi_I (X))$, while from the second we find $$\gamma_m = (W_{\partial\partial}^{(i-1,j)})^{-1}(\pi_I(-e_m (W_{\partial\partial}^{(i,j)})^{-1}(\pi_I (X)) + X_m)),$$ where $\pi_I$ stands for the projection to the image of the map $W_{\partial\partial}^{(i,j)}$.
\end{proof}

\

 \begin{remark}
One has to be careful here because the map $W_{\partial\partial}^{(i,j)}$ can be noninvertible. Hence technically here we are finding the values of $\gamma$ and $\gamma_m$ up to terms in the kernel of the map $W_{\partial\partial}^{(i,j)}$, and we need to keep using the projection $\pi_I$ at all times.
\end{remark}
As an example, we consider  $Q^\partial \omega$: it contains a term of the form $\lambda (W_{\partial}^{(1,2)})^{-1}(\epsilon_n F_\omega)$. Here, $X= \epsilon_n F_\omega$. Hence, we  have
\begin{align*}
\widetilde{Q}^{\partial\partial} \omega &= \dots +  (W_{\partial\partial}^{(
1,2)})^{-1}(\epsilon_n F_\omega), \\
\widetilde{Q}^{\partial\partial} \omega_m &= \dots + (W_{\partial\partial}^{(
0,2)})^{-1}(\pi_I(-e_m (W_{\partial\partial}^{(
1,2)})^{-1}(\epsilon_n F_\omega) + \epsilon_n F_{\omega_m})) + K,
\end{align*}
where $e  K=0$.
Notice that since $W_{\partial\partial}^{(1,2)}$ is surjective on $\Omega_{\partial\partial}^{1,2}$, we do not need the projection on $\epsilon_n F_\omega$, while, since the map $W_{\partial\partial}^{(
0,2)}$ is neither surjective nor injective on $\Omega_{\partial\partial}^{0,2}$, we need the projection $\pi_I$ on the second expression and we still miss something in the kernel of $W_{\partial\partial}^{(
0,2)}$, denoted by $K$.

A similar procedure is needed also for $Q^{\partial}{y^{\dag}}$. On the boundary we have
$$ \widetilde{e}_i \widetilde{Q^{\partial}{y^{\dag}}}= \lambda \widetilde{\sigma_i} \widetilde{y}^{\dag} + \widetilde{\mu}\widetilde{\gamma_i}^\dag$$ for $i = a, m$.
Hence, since $y_m^\dag$ is a top form on the boundary, we get
\begin{align*}
e_m Q^{\partial}{y^{\dag}_m} \dd x^m &= \lambda \sigma_m y^\dag_m \dd x^m + \mu_m \dd x^m \gamma_m^{\dag}, \\
e_a Q^{\partial}{y^{\dag}_m} \dd x^m &= \lambda \sigma_a y^\dag_m \dd x^m + \mu  \gamma_{am}^{\dag} \dd x^m,
\end{align*}
from which we can easily deduce the expression of $\widetilde{Q}^{\partial}$ on the pre-corner.

  \begin{refcontext}[sorting=nyt]
    \printbibliography[] 
  \end{refcontext}

\end{document}